\documentclass[journal]{IEEEtran}
\usepackage{multicol}
\usepackage {graphicx}
\usepackage[dvips]{epsfig}
\usepackage{cite}
\usepackage{color}
\usepackage{subfigure}
\usepackage{amssymb}
\usepackage{amsthm}
\usepackage{amsmath}
\usepackage{leftidx}
\usepackage{mathtools}  
\usepackage{amsfonts}
\newtheorem{theorem}{Theorem}

\newtheorem{remark}{Remark}
\newtheorem{lemma}{Lemma}
\newtheorem{proposition}{Proposition}

\def\rea{\mathbb{R}}
\usepackage{multicol,lipsum}
\begin{document}
\title{Adaptive Motion Control of Parallel Robots
with Kinematic and Dynamic Uncertainties}
\author{
M.~Reza~J.~Harandi, S.~A.~Khalilpour,
  Hamid.~D.~Taghirad,~\IEEEmembership{Senior Member,~IEEE,} and Jose~Guadalupe~Romero
  \thanks{M.~Reza~J.~Harandi, S.~A.~Khalilpour and
    Hamid.~D.~Taghirad are with \textbf{A}dvanced
    \textbf{R}obotics and \textbf{A}utomated \textbf{S}ystems (ARAS),
    Faculty of Electrical Engineering,
    K. N. Toosi University of Technology,
    P.O. Box 16315-1355, Tehran, lran.
    (email: \{Jafari, khalilpour\}@email.kntu.ac.ir and
    taghirad@kntu.ac.ir)}
\thanks{Jose~Guadalupe~Romero is with the Departamento Acad\'{e}mico de Sistemas Digitales, ITAM, R\'io Hondo 1, Progreso Tizap\'an, 01080, Ciudad de M\'exico, M\'{e}xico. (email: jose.romerovelazquez@itam.mx)}
}
 \maketitle
\begin{abstract}
One of the most challenging issues in adaptive control of robot
manipulators with kinematic uncertainties is requirement of the
inverse of Jacobian matrix in regressor form. This requirement  is
inevitable in the case of the control of parallel robots, whose
dynamic equations are written directly in the task space. In this
paper, an adaptive controller is designed for parallel robots based
on representation of Jacobian matrix in regressor form, such that
asymptotic trajectory tracking is ensured. The main idea is separation of determinant and adjugate of Jacobian matrix and then organize new regressor forms.
 Simulation and experimental results on a 2--DOF
R\underline{P}R and 3--DOF redundant cable driven robot, verify
promising performance of the proposed methods.
\end{abstract}
\section{Introduction}
Uncertainties  in dynamic and kinematic parameters are inseparable
part of robotic systems. To design effective controllers in presence
of uncertainty, several methods are reported in the literature. One
of the powerful methods is adaptive
control~\cite{ioannou2012robust}. Adaptive controllers are developed
to dispel dynamic uncertainties in both
serial~\cite{ortega1989adaptive} and parallel robots~\cite{
xia2018modeling}. The main idea in this method is to express dynamic
formulation in regressor form, and furthermore, to derive an
adaptation law for unknown parameters based on a suitable Lyapunov
analysis~\cite{slotine1987adaptive}. In this regard, the first
Jacobian adaptation algorithm for serial robots was presented
in~\cite{cheah2004approximate}, where the velocity equations of the
robot was expressed in regressor form with respect to unknown
kinematic parameters. By using Lyapunov direct method, it is shown
that task space variables track the desired trajectory, whereas
parameters estimations may not necessarily converge to their real
values \cite{cheah2006adaptive}. Note that in this work it is
assumed that an equation containing inverse of Jacobian matrix can
be expressed in regressor form. Wang in~\cite{wang2016adaptive}
resolve this problem and proposed a new adaptation law which
improved the performance of the closed-loop system. However, these
works are focused on serial robots, and less attention has been paid
to the control of parallel robots with Jacobian and kinematic
uncertainties.

Parallel robots are closed--loop mechanisms in which the moving
platform is linked to the base by several independent kinematic
chains~\cite{taghirad2013parallel}. The unique characteristics of
parallel robots in terms of their speed and rigidity make them
suitable to a variety of applications such as flight simulators and
very fast pick and place manipulators \cite{harandi2017motion}.
Cable driven parallel robots are a prominent class of this robots
where the links are formed by cables driven by
actuators~\cite{oh2005reference}. Sincedynamics formulation of these
robots are usually written in task
space~\cite{taghirad2013parallel}, exact values of dynamic
parameters and Jacobian matrix is required to achieve a precise
trajectory tracking. This condition may not be satisfied in most
cases especially for deployable cable driven robots, where a
calibrated model is usually unavailable~\cite{khalilpour2019robust}.

Although calibration methods are well developed to reduce kinematic
uncertainties~\cite{chen2018complete}, they usually does not
overcome Jacobian uncertainties and are not applicable for special
cases such as deployable cable driven robots. Another crucial issue
in large scale cable driven robot is sagging of
cables~\cite{meunier2009control}. In this situation, the kinematic
and Jacobian matrix are changed based on position of end-effector,
and therefore, a control strategy to adapt kinematic and Jacobian is
strictly required. Authors in \cite{babaghasabha2015adaptive,
babaghasabha2015robust} have taken two approaches to tackle this
problem for a specified robot. In~\cite{babaghasabha2015adaptive},
an adaptive controller is proposed, and it is assumed that the
adapted parameters converge to their physical values. This
assumption is not necessarily fulfilled in practice since there is
no theoretical guarantee for such convergence.
In~\cite{babaghasabha2015robust} an adaptive robust controller is
proposed, in which the bounds of dynamic and kinematic estimation
errors are considered to be constant but unknown and at last an
ultimate bound for tracking error is derived. However, since these
bounds are state-dependent, this assumption may not be easily
fulfilled.

In this paper an adaptive controller based on Slotine
and Li method \cite{slotine1987adaptive}, is developed for parallel
manipulators with kinematic and dynamic uncertainties. The proposed
method works well for both fully and redundantly actuated robots.
Invoking the researches in the field of serial robots, the main
contribution of this paper is based on a novel representation of
Jacobian matrix of the robot in a general regression form, i.e.
instead of expressing velocity terms in regressor form, the Jacobian
matrix is represented in regressor form which clearly result in a
matrix of unknown values. In order to rectify expression of the
inverse of Jacobian matrix in regressor form, which is a necessary
part of control law and  it is also a stumbling barrier in all the
previous works on serial robots, we separate adjugate and
determinant of this matrix to form new regressors. Finally, based on
passivity method, trajectory tracking is analyzed using direct
Lyapunov method. Note that this is, to the best of the authors'
knowledge, not fully addressed before in the field of parallel
robots with detailed analysis.

\textbf{Notation}: For any matrix $A\in \mathbb{R}^{n\times m}$,
$A_i$ denotes $i$-th column, $\prescript{}{j}{A}$ denotes $j$-th row
and $A_{i,j}$ denotes $(i,j)$-th element of $A$. $A^*$ and
$A^\dagger$ represent adjugate and right pseudo-inverse of $A$,
respectively, while $\hat{A}$ represents estimated value of $A$ and
$\tilde{A}=\hat{A}-A$. Unless indicated otherwise, all vectors in
the paper are considered as column vectors.

\section{Kinematics and Dynamics Analysis}\label{s1}
The dynamic model of a parallel robot with $n$ degrees of freedom
and $m$ actuators with negligible  dissipation forces may be written
in the task space as follows~\cite{khalilpour2019robust}:
\begin{equation}
\label{1} M(X)\ddot{X}+C(X,\dot{X})\dot{X}+G(X)=F=J^T(X)\tau,
\end{equation}
where $X, \dot X \in \mathbb{R}^n$ denotes the generalized
coordinate vector representing the position and orientation of the
end--effector and their velocities, respectively, $\tau\in
\mathbb{R}^m$ denotes the applied torque to the robot, $M(X)\in
\mathbb{R}^{n\times n}$ is the inertia matrix, $C(X,\dot{X})\in
\mathbb{R}^{n\times n}$ denotes the Coriolis and centrifugal matrix,
$G(X)\in \mathbb{R}^n$ is the vector of gravity terms, $J(X)\in
\mathbb{R}^{m\times n}$ denotes the Jacobian matrix of the robot.
Some important properties of the robot dynamic formulation~(\ref{1})
from~\cite[Sec. 5.5.4]{taghirad2013parallel} are as follows.
\begin{itemize}
\item[\textbf{P1}:] The inertia matrix $M(X)$ is symmetric and positive definite for all $X$.
\item[\textbf{P2}:] The matrix $\dot M(X)-2C(X,\dot{X})$ is
skew symmetric.
\item[\textbf{P3}:] The dynamic model is linear with
respect to a set of dynamical parameters and may be represented in a
linear regression form:
\end{itemize}
\begin{equation}
\label{2}
M(X)\ddot{X}+C(X,\dot{X})\dot{X}+G(X)=Y_m(\ddot{X},\dot{X},X)\theta_m,
\end{equation}
where, $Y_m(\ddot{X},\dot{X},X)$ denotes the regressor matrix and
$\theta_m$ denotes the dynamic parameters vector.

The task space wrench $F$ is related to joint space force vector
$\tau$ by Jacobian transpose:
\begin{equation}
\label{3}
F=J^T(X)\tau.
\end{equation}

It was shown in~\cite{cheah2004approximate} that for serial robots,
the Jacobian matrix may be expressed in regressor form as:
\begin{equation}
\label{4} J(q)\dot{q} =Y_k(q,\dot{q})\theta_k,
\end{equation}
where $\theta_k$ denotes unknown kinematic parameters in Jacobian matrix.
This expression may be used to represent each element of the Jacobian
matrix as a linear regression form of kinematic parameters as:
\begin{equation}
\label{5}
J_{i,j}^T(q)=Y_{k_i}\theta_{k_j}.
\end{equation}
Thus, one may represent $J^T(q)$ as follows:
\begin{equation}
\label{6} J^T(q)=\displaystyle\sum_{i=1}^{n}
\displaystyle\sum_{j=1}^{m} Y_{k_i}\theta_{k_j}\prescript{i}{}\Psi^j,
\end{equation}
where all elements of $\prescript{i}{}\Psi^j \in \mathbb{R}^{n\times
m}$ are zero except  $(i,j)$-th element which is equal to one.

In parallel robots with actuated revolute joints, Jacobian matrix is
expressible in the form of (\ref{6}). However, Jacobian matrix of
actuated prismatic joints including cable driven robots may be
represented as,~\cite[Ch.4]{taghirad2013parallel}:
\begin{equation}
\label{3a}
J(X)=
\begin{bmatrix}
\hat{\lambda}_1^T & (\prescript{b}{}{R}_p a_1\times\hat{\lambda}_1)^T  \\
\vdots & \vdots  \\
\hat{\lambda}_m^T & (\prescript{b}{}{R}_p a_m\times\hat{\lambda}_m)^T
\end{bmatrix}
\end{equation}
where $\hat{\lambda}_i$ denotes unit vector in opposite side of
link's direction, $a_i$ denotes the attachment points of the links
to the end-effector represented in moving frame and
$\prescript{b}{}{R}_p$ denotes the rotation matrix. On the contrary,
it is not straight forward for these manipulators to express $J(X)$
in form of (\ref{6}) due to fractional elements of the matrix. To
overcome this problem, $J^T(X)$ is expressed in the following form:
 \begin{align}
 &\resizebox{.99\hsize}{!}{$J^T=
\begin{bmatrix}
{\lambda}_1 & \dots & {\lambda}_m \\
(\prescript{b}{}{R}_p a_1\times{\lambda}_1)  & \dots &
(\prescript{b}{}{R}_p a_m\times{\lambda}_m)
\end{bmatrix}
\begin{bmatrix}
l_1 & \dots & 0 \\
\vdots & \ddots & \vdots \\
0 & \dots & l_m
\end{bmatrix}^{-1}$}\nonumber\\&
\triangleq J^T_{new}(X)L^{-1},
 \label{6a}
\end{align}
where $\lambda_i=l_i\hat{\lambda_i}$ and $l_i$ as the
length of $i$-th link. Through this transformation, it is possible
to define $J^T_{new}(X)\in \mathbb{R}^{n\times m}$ in the regressor form of
(\ref{6}). Invoking (\ref{6}), let us write $J^T_{new}(X)$ in the following compact form
\begin{equation}
\label{7} J_{new}^T(X)=Y(X)\Theta,
\end{equation}
in which $Y(X)\in \mathbb{R}^{n\times l}$ and $\Theta\in \mathbb{R}^{l\times m}$. We can notice that it is possible to show that this representation is general and does not assign merely to Jacobian matrix in the form (\ref{3a}).


\section{Adaptive Jacobian Controller}\label{s2}
In this section an adaptive controller based on Slotine and Li method is proposed for
a parallel manipulator with uncertain kinematics and dynamics. It is
assumed that position and velocity of end-effector, as well as the
length of links for the robots are available for feedback and
derivation of Jacobian matrix in the form of (\ref{3a}). In the
proposed controller, trajectory tracking is guaranteed by
combination of Slotine and Li controller, and adaptation law for the
unknown parameters.

Let us define $S$ as~\cite{
babaghasabha2015robust,slotine1987adaptive}:
\begin{equation}
\label{8} S=\dot{\tilde{X}}+\Gamma\tilde{X}=\dot{X}-\dot{X}_r,
\end{equation}
with
\begin{equation}
\label{9} \tilde{X}=X-X_d, \qquad
\dot{X}_r=\dot{X}_d-\Gamma\tilde{X},
\end{equation}
where $X_d\in \mathcal{C}^2$ denotes the desired trajectory,
$X_r=X_d-\Gamma\int_{0}^{t}\tilde{X} dt$ denotes virtual reference
trajectory and $\Gamma$ is a constant positive definite matrix. If
all the kinematics and dynamics parameters are known, the following
control law may be directly used for a suitable performance
requirement
\begin{equation}
\label{10} \tau=LJ_{new}^\dagger
\big(M\ddot{X}_r+C\dot{X}_r+G-KS\big),
\end{equation}
where, $K$ is constant symmetric positive definite matrix, and
$J_{new}^\dagger$ denotes the right pseudo-inverse of
$J^T_{new}(X)$. In the case of fully actuated robots,
$J_{new}^\dagger$ is replaced by $J_{new}^{-T}$. Note that control
law (\ref{10}) is related to Jacobian matrix in the form (\ref{6a}),
while for actuated revolute joint, the control law is
    \begin{equation*}
 \tau=J^\dagger
    \big(M\ddot{X}_r+C\dot{X}_r+G-KS\big).
    \end{equation*}
In the sequel, we continue with the notation (\ref{10}). Let us
write $J_{new}^\dagger=\frac{R}{T}$, for the case of redundantly
actuated robot, these matrices are defined as
\begin{equation}
\label{11} R=J_{new}(J_{new}^TJ_{new})^* \in \rea^{m \times n}, \hspace{4mm} T=\det (J_{new}^TJ_{new}) \in \rea,
\end{equation}
and for the case of fully parallel robots, i.e. the robots with number of actuators  equal to degrees
of freedom,
 \begin{equation}
\label{12} R=(J_{new}^T)^* \in \rea^{n\times n}, \hspace{3mm} T=\det (J_{new}^T) \in \rea
\end{equation}
where $(\cdot)^*$ denotes the adjugate matrix. Due to the
uncertainties in parameters, we have to use the estimated values in
the control law
\begin{equation}
\label{13} \tau=L\frac{\hat{R}}{\hat{T}}
\big(\hat{M}\ddot{X}_r+\hat{C}\dot{X}_r+\hat{G}-KS\big),
\end{equation}
where, $\hat{(\,\cdot\,)}$ denotes the estimated value. Invoking
\textbf{P3} and this fact that adjugate matrix is linear with
respect to the parameters, we may express
 \begin{equation}
\label{14}
\begin{array}{c}
\hat{R}\big(\hat{M}\ddot{X}_r+\hat{C}\dot{X}_r+\hat{G}-
KS\big)=Y_a(X,\dot{X},\dot{X}_r,\ddot{X}_r)\hat{\theta}_a,
\end{array}
\end{equation}
where $\theta_a\in \mathbb{R}^r$ is constructed by concatenation of
kinematics and dynamics parameters. Using the proposed control law,
the closed-loop dynamics may be written as:
 \begin{equation}
\label{15} M(X)\ddot{X}+C(X,\dot{X})\dot{X}+G(X)=
J_{new}^T\frac{Y_a\hat{\theta}_a}{\hat{T}}.
\end{equation}
Adding $$-J_{new}^T\frac{Y_a\theta_a}{\hat{T}}=
-J_{new}^T\frac{T}{\hat{T}}\frac{Y_a\theta_a}{T}=
\frac{T}{\hat{T}}\big(M\ddot{X}_r+C\dot{X}_r+G-KS\big)$$ to both
sides of (\ref{15}), the following equation is obtained:
 \begin{align}
M(X)\ddot{X}+C(X,\dot{X})\dot{X}+G(X)-
\frac{T}{\hat{T}}\big(M\ddot{X}_r&\nonumber\\+C\dot{X}_r
+G-KS\big)=J_{new}^T(X)\frac{Y_a\tilde{\theta}_a}{\hat{T}},&
\label{16}
\end{align}
where $\tilde{\theta}_a=\hat{\theta}_a-\theta_a$ denotes estimation error.
Determinant is linear with respect to the elements of the matrix,
thus we may express $T$ as a linear regression $T=Y_b(X)\theta_b$
where $\theta_b\in\mathbb{R}^k$ are unknown parameters in
determinant. On the other hand, considering \textbf{P3}, one may
reach to the following equation:
\begin{equation}
\label{17} M\ddot{X}_r+C\dot{X}_r+G-KS=
Y_c(X,\dot{X},\dot{X}_r,\ddot{X}_r)\theta_c,
\end{equation}
where $\theta_c\in\mathbb{R}^p$ denotes the vector of dynamical
parameters. Using (\ref{17}), left hand side of (\ref{16}) is
rewritten as follows:
  \begin{align}
&M(X)\ddot{X}+C(X,\dot{X})\dot{X}+G(X)-\nonumber\\
&\frac{Y_b(\hat{\theta}_b -\tilde{\theta}_b)}
{\hat{T}}\big(M\ddot{X}_r+C\dot{X}_r +G -KS\big)= \nonumber\\
&M\dot{S}+CS+KS+
\frac{Y_b\tilde{\theta}_b}{\hat{T}}\big(M\ddot{X}_r+
C\dot{X}_r+G-KS\big)= \nonumber\\&
M\dot{S}+CS+KS+\frac{Y_b\tilde{\theta}_b}{Y_b\hat{\theta}_b}Y_c\theta_c.
  \label{18}
\end{align}
Finally, using (\ref{7}), closed-loop equation (\ref{16}) yields
\begin{equation}
 \begin{array}{l}
\label{20} M\dot{S}+CS+KS= \\
Y\hat{\Theta}\frac{Y_a\tilde{\theta}_a}
{Y_b\hat{\theta}_b} -Y\tilde{\Theta}\frac{Y_a\tilde{\theta}_a}
{Y_b\hat{\theta}_b}+\frac{Y_b\tilde{\theta}_b}
{Y_b\hat{\theta}_b}Y_c\tilde{\theta}_c-
\frac{Y_b\tilde{\theta}_b}{Y_b\hat{\theta}_b}Y_c\hat{\theta}_c.
\end{array}
\end{equation}
 \begin{figure}[t]
    \centering
   \includegraphics[scale=.8]{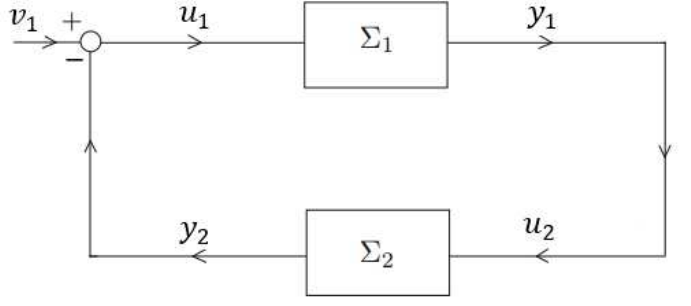}
    \caption{Standard feedback
        configuration of two systems }
    \label{con}
 \end{figure}

In the sequel, we design adaptation laws based on passivity method.
First, recall Proposition 4.3.1 of  \cite{van2000l2} on the
connection of two passive systems. The reader is referred to this
reference for detailed proof.
 \begin{proposition}\label{th2}
Consider the standard feedback closed-loop system $\Sigma_1,
\Sigma_2$ which is shown in Fig.~\ref{con}. Assume that $\Sigma_1$
is output strictly passive, i.e. there exists a storage function $H_1$ such that
$\dot{H}_1\leq u_1^Ty_1-y_1\psi(y_1)$
where $y_1\psi(y_1)\geq 0$,
 and $\Sigma_2$ is passive, i.e.  there exists a storage function $H_2$ such that
 $\dot{H}_2\leq u_2^Ty_2.$
  Then the states of
$\Sigma_1$ converge to zero while states of $\Sigma_2$ remain bounded.
 \end{proposition}
In order to use the above proposition, we shall modify right hand
side of (\ref{20}) in such a way that all the terms are represented
by a regressor matrix and an unknown vector. Therefore, in the
following, $Y\tilde{\Theta}\frac{Y_a\tilde{\theta}_a}
{Y_b\hat{\theta}_b}$ and $\frac{Y_b\tilde{\theta}_b}
{Y_b\hat{\theta}_b}Y_c\tilde{\theta}_c$ are changed accordingly. Assume that
  \begin{equation*}
\frac{Y_b\tilde{\theta}_b}
{Y_b\hat{\theta}_b}Y_c\tilde{\theta}_c\triangleq\frac{Y_cY_\mu}{Y_b\hat{\theta}_b}\tilde{\theta}_{\mu}
\end{equation*}
with,
  \begin{equation*}
\begin{array}{c}
\resizebox{.97\hsize}{!}{$Y_\mu=\begin{bmatrix}
Y_b & 0_{1\times k} & \dots & 0_{1\times k} \\
0_{1\times k} & Y_b & 0_{1\times k} & \dots \\
\vdots & & \ddots & \vdots \\
0_{1\times k} & \cdots & 0_{1\times k} & Y_b
\end{bmatrix}_{n\times p\cdot k} ,
\tilde{\theta}_{\mu}=\begin{bmatrix}
\Big(\prescript{}{1}{(\tilde{\theta}_c\tilde{\theta}_b^T)}\Big)^T \\ \vdots \\ \Big(\prescript{}{p}{(\tilde{\theta}_c\tilde{\theta}_b^T)}\Big)^T
\end{bmatrix}$}
\end{array}
    \end{equation*}
and $\prescript{}{i}{(\tilde{\theta}_c\tilde{\theta}_b^T)}$ is the $i$-th row
of $\tilde{\theta}_c\tilde{\theta}_b^T$. Furthermore,
\begin{equation*}
    Y\tilde{\Theta}\frac{Y_a\tilde{\theta}_a}
    {Y_b\hat{\theta}_b}=\frac{YY_\eta}{Y_b\hat{\theta}_b}\tilde{\theta}_\eta
        \end{equation*}
        where,
     \begin{equation*}
        \begin{array}{c}
        Y_\eta=\\
        \hspace{-2mm}\resizebox{.97\hsize}{!}{$\begin{bmatrix}
        \prescript{}{1}{(Y_a)}, ..., \prescript{}{m}{(Y_a)} & 0_{1\times m.r}& \dots& 0_{1\times m\cdot r} \\
        0_{1\times m\cdot r} & \prescript{}{1}{(Y_a)}, ..., \prescript{}{m}{(Y_a)} & & 0_{1\times m\cdot r} \\ \vdots & & \ddots & \vdots \\
        0_{1\times m\cdot r} & \dots &  & \prescript{}{1}{(Y_a)}, ..., \prescript{}{m}{(Y_a)}
        \end{bmatrix}_{l\times m\cdot r\cdot l}$}
        \end{array}
\end{equation*}
in which $\prescript{}{i}{(Y_a)}$ is $i$-th row of $Y_a$, and
     \begin{equation*}
    \begin{array}{c}
    \tilde{\theta}'_\eta=
    \begin{bmatrix}
    \tilde{\Theta}_{1,1}\tilde{\theta}^T_a &  \dots & \tilde{\Theta}_{1,m}\tilde{\theta}^T_a \\
    \vdots & \ddots & \vdots \\
    \tilde{\Theta}_{l,1}\tilde{\theta}^T_a & \dots & \tilde{\Theta}_{l,m}\tilde{\theta}^T_a
    \end{bmatrix}_{l\times m\cdot r},
    \tilde{\theta}_\eta=\begin{bmatrix}
    \prescript{}{1}{(\tilde{\theta}'_\eta)}^T \\ \vdots \\ \prescript{}{l}{(\tilde{\theta}'_\eta)}^T
    \end{bmatrix}_{m\cdot r\cdot l\times 1}
    \end{array}
    \end{equation*}
where $\tilde{\Theta}_{i,j}$ is $(i,j)$-th element of
$\tilde{\Theta}$. Therefore,  (\ref{20}) may be rewritten as
follows:
    \begin{equation}
    \begin{array}{c}
    \label{29.6} M\dot{S}+CS+KS=\\
    \frac{Y\hat{\Theta}Y_a}{Y_b\hat{\theta}_b}\tilde{\theta}_a
    -\frac{YY_\eta}{Y_b\hat{\theta}_b}\tilde{\theta}_\eta+\frac{Y_cY_\mu}{Y_b\hat{\theta}_b}\tilde{\theta}_{\mu}-
    \frac{Y_c\hat{\theta}_cY_b}{Y_b\hat{\theta}_b}\tilde{\theta}_b\triangleq Y_F\tilde{\theta}_F,
    \end{array}
    \end{equation}
    where,
    $$F=\begin{bmatrix}
    \frac{Y\hat{\Theta}Y_a}{Y_b\hat{\theta}_b} & -\frac{YY_\eta}{Y_b\hat{\theta}_b} & \frac{Y_cY_\mu}{Y_b\hat{\theta}_b} & -
    \frac{Y_c\hat{\theta}_cY_b}{Y_b\hat{\theta}_b}
    \end{bmatrix}$$
    $$\tilde{\theta}_F=\begin{bmatrix}
    \tilde{\theta}_a \\ \tilde{\theta}_\eta \\ \tilde{\theta}_{\mu} \\ \tilde{\theta}_b
    \end{bmatrix}  \in \rea^{(r+m\cdot r\cdot l+p\cdot k+k)\times 1}.$$
Equation (\ref{29.6}) may be considered as the system $\Sigma_1$
represented in Proposition \ref{th2}. This system is output strictly
passive with  $H_1=\frac{1}{2}S^TMS$  as  the storage function,
because
        \begin{equation*}
    \dot{H}\leq -\lambda_{min}\{K\}\|S\|^2+S^Tu_1
        \end{equation*}
with  $u_1=Y_F\tilde{\theta}_F, y_1=u_2=S$ and without external
input (i.e. $v_1=0$). In order to apply Proposition \ref{th2}, an
adaptation law is required to be defined for $\tilde{\theta}_F$ such
that $\Sigma_2$ becomes passive. For this means, the following
dynamic is set for $\tilde{\theta}_F$
    \begin{equation}
            \label{29.9}
        \dot{\hat{\theta}}_F=-\Lambda^{-1}Y^T_FS \hspace{1.5cm} \Lambda> 0
    \end{equation}
which leads to the passivity of $\Sigma_2$, since
    \begin{equation*}
    H_2=\frac{1}{2}\tilde{\theta}^T_F\Lambda\tilde{\theta}_F \to \dot{H}_2=-S^TY_F\tilde{\theta}_F=u_2^Ty_2.
    \end{equation*}
Note that it is assumed that $\theta_F$ is constant. Let us state
the following theorem on adaptive passivity based control of
parallel robots with kinematics and dynamics uncertainties.
    \begin{theorem}\label{th3}
Consider a parallel robot with dynamic equation (\ref{1}), control law (\ref{13}) and adaptation
law (\ref{29.9}). The tracking error converges to zero in the
presence of uncertainties in kinematics and dynamics
parameters.
    \end{theorem}
\begin{proof}
The proof is obvious with respect to Proposition \ref{th2} and
$\Sigma_1, \Sigma_2$ defined above. However, a Lyapunov  based proof
is also presented here. Consider the following Lyapunov function
candidate:
        \begin{equation}
    \label{29.8}
     V=\frac{1}{2}S^TMS+\frac{1}{2}\tilde{\theta}^T_F\Lambda\tilde{\theta}_F.
        \end{equation}
      then its
 time derivative becomes
\begin{equation*}
 \dot{V}\leq -S^TKS.
    \end{equation*}
Invoking Lasalle-Yoshizawa Theorem \cite[Theorem
8.4]{khalil2002nonlinear},  it is easy to show that $S$ converges to
zero, and hence, the convergence of $\tilde{X}$ is resulted from
(\ref{8}).
\end{proof}

\begin{remark}
\normalfont There may be a number of parameters in multiple unknown
vectors that may not converge to their real values. However, this
does not cause any problem for a suitable trajectory tracking.
\end{remark}
\begin{remark}\label{r1}
    \normalfont
Singularity avoidance in construction and path planing is a
necessary and important requirement in parallel robots
\cite{gosselin2016kinematically}. Here, it is assumed that desired
trajectory is inside its workspace far from singular space of the
robot. By this means, the Jacobian matrix is always full rank, and
therefore, its estimation is plausible. However, projection
algorithm may be employed in order to ensures singularity avoidance
as well as avoiding large variation in parameters and provides a
faster and better transient response. Note that by this means,
positive tension in the case of cable driven robots is ensured.
\end{remark}
In the following lemma, invoking \cite[Theorem
4.4.1]{ioannou2012robust}, a projection algorithm based on gradient
method is proposed.
\begin{lemma}\label{th4}
Consider closed-loop system (\ref{29.6}) with adaptation law
(\ref{29.9}). Assume that it is priori known that ${\theta}_F$ is
absolutely in a compact subspace $\Omega$, i.e. ${\theta}_F \in
\Omega$ where $\Omega$ is defined as $\Omega=\{{\theta}_F |
g({\theta}_F)\leq 0\}$ and $g({\theta}_F)$ is known. The objective
is to keep $\hat{\theta}_F$ in $\Omega$.
If $\hat{\theta}_F$ is on the edge of $\Omega$ i.e.
$\hat{\theta}_F\in \partial \Omega$, and
$\dot{\hat{\theta}}_F^T\nabla g>0$,  the following adaptation law
is chosen
    \begin{equation}
    \label{29.91}
\dot{\hat{\theta}}_F=-\Lambda^{-1}\overline{\nabla
    g}Y_F^T S
    \end{equation}
where $\overline{\nabla g}$ is projection matrix
        \begin{equation}
    \label{29.92}
\overline{\nabla g}=I-\frac{(\nabla g)(\nabla g)^T}{||\nabla g||^2} .
    \end{equation}
This leads to $\hat{\theta}_F$ to remains in $\Omega$.
\end{lemma}
\begin{proof}
If $\hat{\theta}_F$ is inside $\Omega$, adaptation law (\ref{29.9})
is applied, and by this means, it remains in $\Omega$. Assume that
${\hat{\theta}_F}\in \partial \Omega$, hence the aim is to ensure
that $\hat{\theta}_F$ always remains in $\Omega$. For this means,
the direction of $\dot{\hat{\theta}}_F$ should not be directed
toward outside of $\Omega$. In other words, dot product of
$\dot{\hat{\theta}}_F$ and $\nabla g=\partial g/\partial \theta$
shall be non-positive. Therefore, if $(-\Lambda^{-1}Y^T_FS)^T \nabla
g> 0$, $\dot{\hat{\theta}}_F$ should be projected on the direction
tangent to $\partial \Omega$. This is done using projection matrix
(\ref{29.92}) which results in adaptation law (\ref{29.91}). Now
consider Lyapunov candidate (\ref{29.8}) whose time derivative is
given as:
\begin{equation*}
\dot{V}=-S^TKS+S^TY_F\tilde{\theta}_F-\tilde{\theta}_F^T\overline{\nabla g}Y_F^T S.
\end{equation*}
By considering (\ref{29.92}), $\dot{V}$ becomes
\begin{equation*}
\dot{V} \leq -\lambda_{min}\{K\}\|S\|^2+\tilde{\theta}_F^T\frac{(\nabla g)(\nabla g)^T}{||\nabla g||^2}Y_F^T S.
\end{equation*}
Note that $\tilde{\theta}_F^T\nabla g\geq 0$, since direction of
$\tilde{\theta}_F$ and $\nabla g$ are toward outside of $\Omega$.
Therefore, the last term in the above inequality is negative.
\end{proof}
Notice that in most cases, exact derivation of $g(\theta_F )$ is
highly complicated. Hence, the acceptable bound for each element of
unknown vector is considered and the simplest projection function,
namely saturation is used, since it is applicable to any adaptive
control law~\cite{dixon2007adaptive}. In other words, this is
equivalent to define an absolute value function for every elements
of unknown vectors. For example, Assume that $\alpha\in \mathbb{R}$
is an element of an unknown vector and it is known that $k_1\leq
\alpha\leq k_2$. Define $g(\alpha)$ as
    $$g(\alpha)=\left |\alpha-\frac{k_1+k_2}{2}\right |-\frac{k_2-k_1}{2}.$$
Now, one may find $\nabla g=\text{sign}(\alpha-\frac{k_1+k_2}{2})$ which leads to $\overline{\nabla g}=0$ and therefore, $\dot{\hat{\alpha}}=0$ when $\alpha$ is at the edge of $g(\alpha)\leq 0$.
\begin{figure}[t]
    \centering
    \includegraphics[scale=.7]{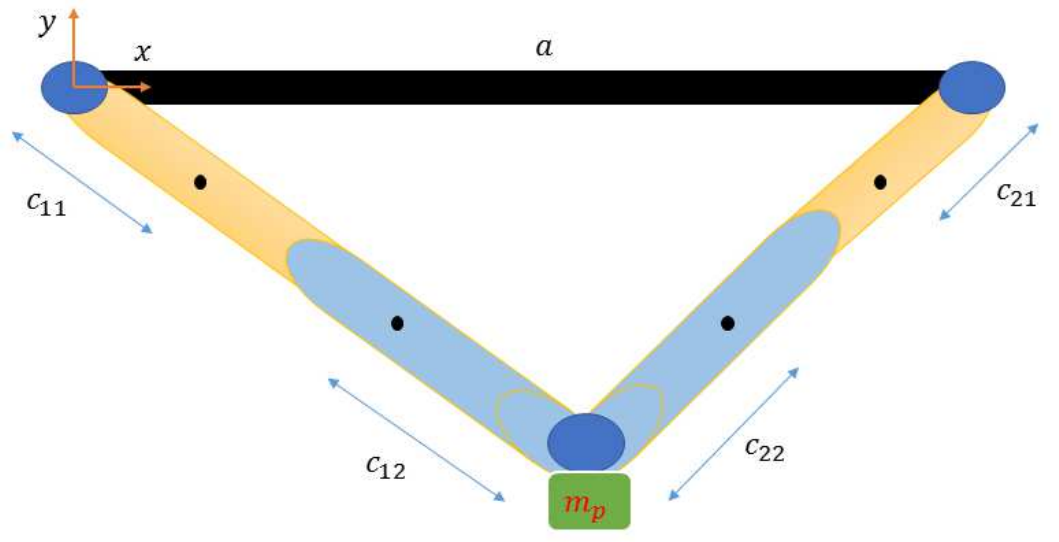}
    \caption{Schematic of the 2--DOF R\underline{P}R parallel robot.  }
    \label{sh0}
\end{figure}
\section{Simulation Results}\label{s2.5}
In this section, simulation results of proposed method on a 2--DOF R\underline{P}R parallel robot is presented. The schematic of this robot is illustrated in Fig.~\ref{sh0}. $X=[x,y]^T$ denotes position of end-effector, $I_{xi}$ is moment of inertial of $i$-th link,
$m_{i1}$ and $m_{i2}$ are the mass of $i$-th cylinder and piston, respectively and $m_p$ denotes the mass of end-effector. The dynamic parameters and Jacobian matrix of the robot are
\begin{align*}
M&=m_pI_2+\displaystyle\sum_{i=1}^{2} m_{i2}\hat{\lambda}_i\hat{\lambda}_i^T-\frac{1}{l_i^2}I_{xi}\hat{\lambda}_{i\times}^2-m_{ce}\hat{\lambda}_{i\times}^2 \\
C&=\displaystyle\sum_{i=1}^{2} -\frac{2}{l_i}m_{co}\dot{l}_i\hat{\lambda}_{ix}^2-\frac{1}{l_i^2}m_{i2}c_{i2}\hat{\lambda}_i\dot{X}^T\hat{\lambda}_{ix}^2\\
G&=\bigg(m_p+\displaystyle\sum_{i=1}^{2}m_{ge}\hat{\lambda}_{ix}^2-m_{i2}\hat{\lambda}_i\hat{\lambda}_i^T\bigg)\begin{bmatrix}
0 \\ g
\end{bmatrix}\\
J&=-\begin{bmatrix}
\frac{x}{l_1} & \frac{y}{l_1} \\
\frac{x-a}{l_2} & \frac{y}{l_2}
\end{bmatrix},
\hat{\lambda}_1=\begin{bmatrix}
\frac{x}{l_1} \\ \frac{y}{l_1}
\end{bmatrix},
\hat{\lambda}_2=\begin{bmatrix}
\frac{x-a}{l_2} \\ \frac{y}{l_2}
\end{bmatrix}
\end{align*}
with
\begin{align*}
l_1^2&=x^2+y^2, \quad l_2^2=(x-a)^2+y^2,\quad\dot{l}_i=J_{i,1}\dot{x}+J_{i,2}\dot{y}\\
m_{ce}&=\displaystyle\sum_{i=1}^{2} \frac{1}{l_i^2}\big(m_{i1}c_{i1}^2+m_{i2}c_{i2}^2\big)\\
m_{co}&=\frac{1}{l_i}m_{i2}c_{i2}-\frac{1}{l_i^2}(I_{xi}+l_i^2m_{ce})\\
m_{ge}&=\frac{1}{l_i}\Big(m_{i1}c_{i1}+m_{i2}(l_i-c_{i2})\Big).
\end{align*}
    \begin{table}[b]
            \caption{Parameters of 2R\underline{P}R robot.}
                \centering
    \large
    \begin{tabular}[b]{|c|c|c|c|c|c|}
        \hline
        & $m_{i1}$ & $m_{i2}$ & $c_{i1}$ & $c_{i2}$ & $I_{xi}$ \\
        \hline
        \text{$i=1$}    &  1& 1 & 0.5 & 0.5 & 0.1 \\
        \hline
            \text{$i=2$}    &  1& 1 & 0.5 & 0.5 & 0.1 \\
            \hline
    \end{tabular}
    \label{t1}
\end{table}
The parameters of the robot are shown in Table~\ref{t1}. The mass of
end-effector is considered equal to 2Kg. All of the regressors are
represented in Appendix.

In order to evaluate performance of proposed method in
Theorem~\ref{th2}, a simulation with adaptive robust controller
proposed in \cite{babaghasabha2015robust} is considered. The
parameters of the robot are perturbed by 25\%. The gains of
controllers are chosen as
\begin{align*}
\Gamma=2I,\qquad K=3I,\qquad \Lambda=5I.
\end{align*}
Simulation results are illustrated in Fig.~\ref{p5}. Configuration
variables of the robot converge to desired values in both methods.
However, the control signal with adaptive robust method has an
undesirable chattering which is not practically acceptable. Note
that as indicated in \cite{babaghasabha2015robust}, it is possible
to avoid chattering with the expense of loosening the asymptotic
stability to UUB tracking error.
\begin{figure}[t!]
    \centering
    \subfigure[Simulation results of the proposed method.]{
        \includegraphics[width=.99\linewidth]{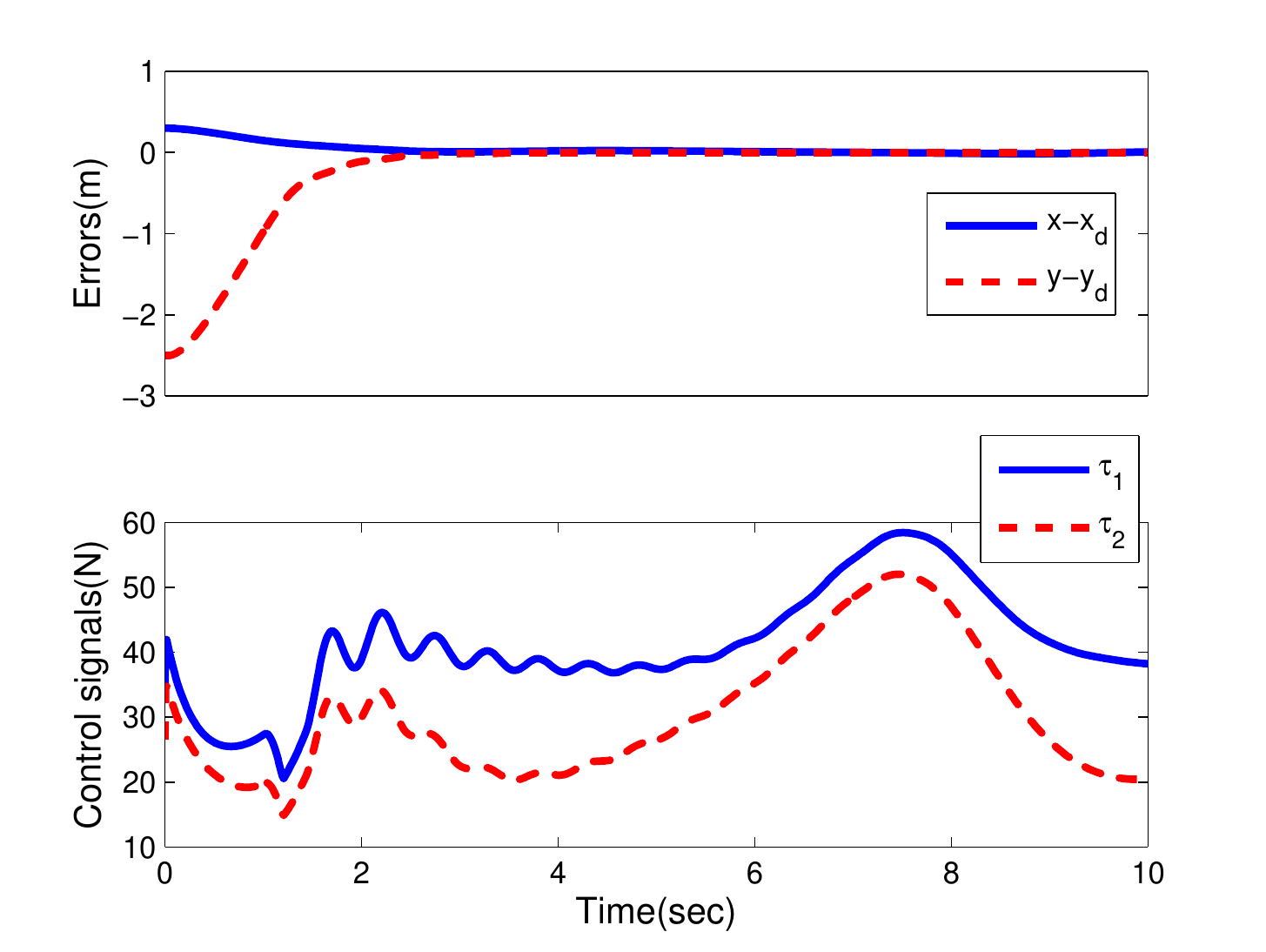}
    }
    \subfigure[Simulation results of adaptive robust
    controller proposed by \cite{babaghasabha2015robust}.]{
        \includegraphics[width=.99\linewidth]{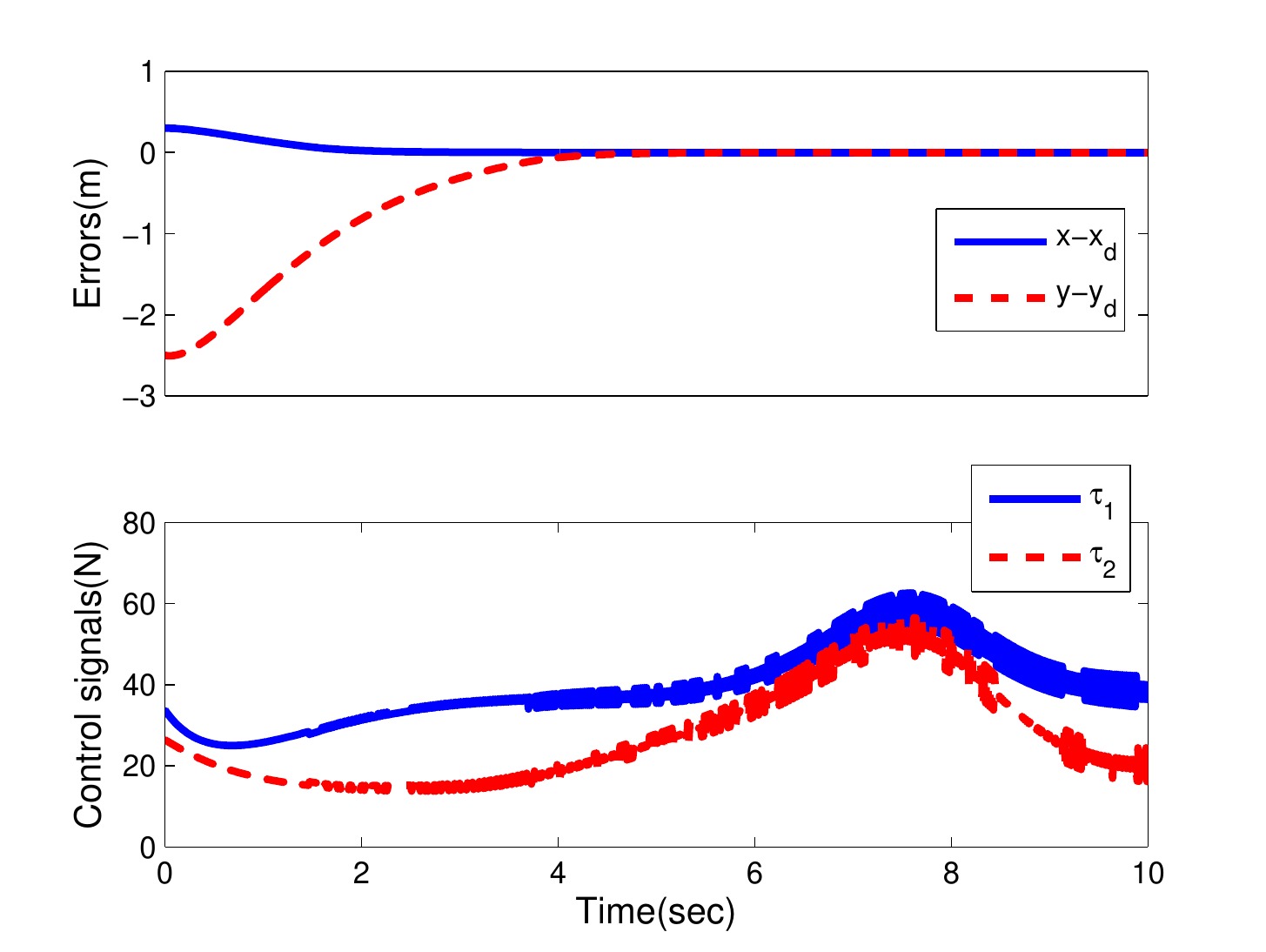}
    }
    \caption{ Tracking error converges to zero with both method while chattering in control law is  destructive inevitable part of adaptive robust controller.}
    \label{p5}
\end{figure}

\section{ Experimental Results}\label{s3}
In order to verify the performance of the proposed method in experiment, a 3--DOF suspended
Cable Driven  Robot (CDR) is considered. The schematic of the
robot is illustrated in Fig.~\ref{sh1}. End-effector is suspended from anchor points by cables which are controlled by motors.
 All of the anchor points are in the same height. The robot has three translational degrees of
freedom with four actuated cables which are driven by motors
through pulleys.
\begin{figure}[t]
    \centering
\includegraphics[scale=.1]{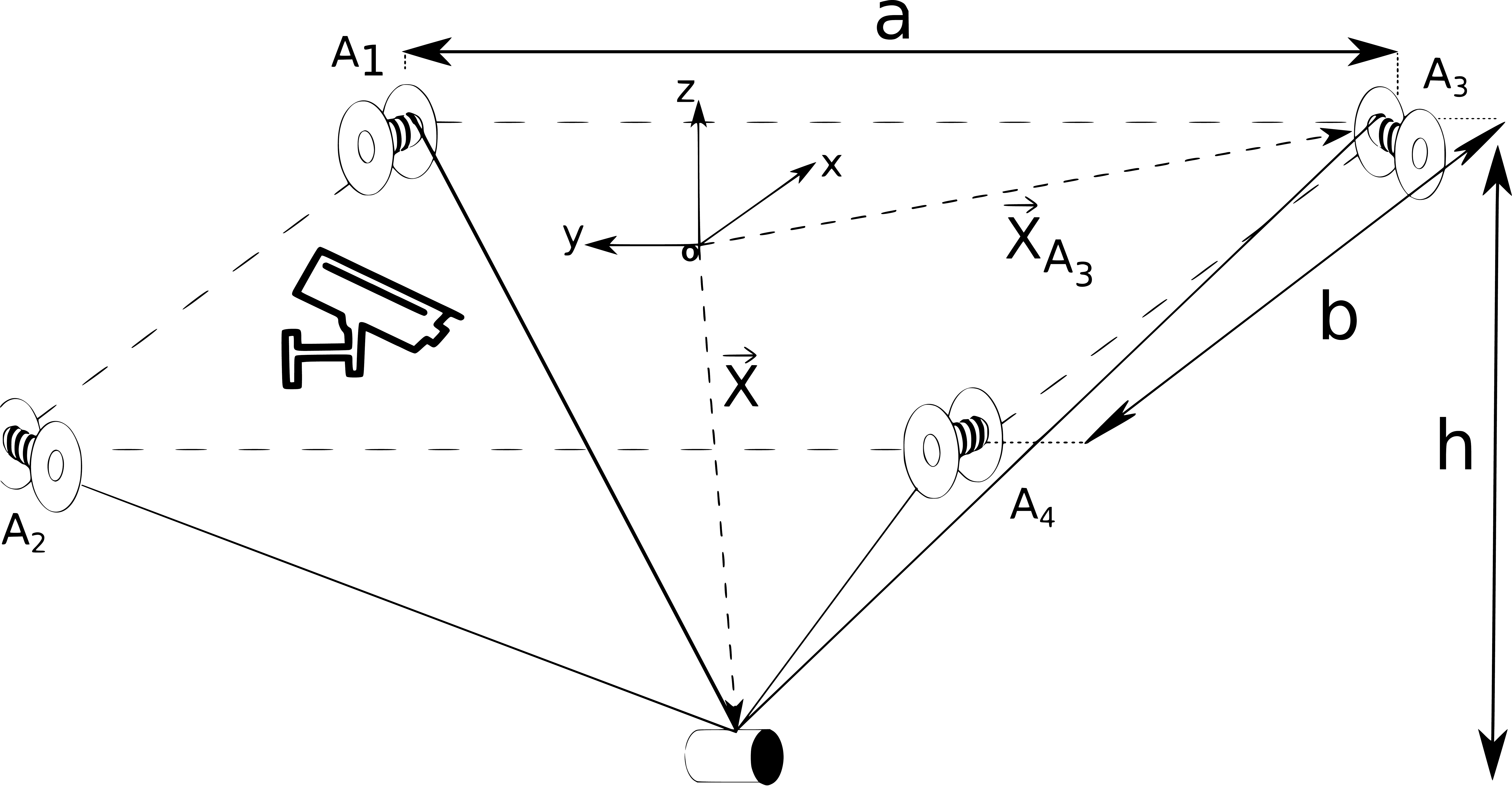}
    \caption{Schematic of the suspended robot with four cables. All of the anchor points are at the same height. Center of coordinates is located in the middle of  $A_1-A_2-A_3-A_4$ rectangle with zero height. }
    \label{sh1}
\end{figure}
Kinematics formulation of this robot is given by
\begin{equation}\label{30-1}
    l_i^2=(x-x_{Ai})^2+(y-y_{Ai})^2+(z-z_{Ai})^2 \quad i=1,...,4
\end{equation}
where, $X=[x,y,z]^T$ is the position of end-effector and   $x_{Ai},y_{Ai},z_{Ai}$ are the uncertain kinematic
parameters that determine the cable anchor points.
Dynamic matrices of the robot with the assumption of massless and infinitely stiff
cables are as follows
\begin{equation}
M=
\begin{bmatrix}
m & 0 & 0\\
0 & m & 0\\
0 & 0 & m
\end{bmatrix}
\quad  C=0_{3\times 3} \quad  G=
\begin{bmatrix}
0 \\ 0 \\ mg
\end{bmatrix}
\end{equation}
where $m$ is the mass of end-effector.

Since the proposed method is also applicable to redundantly actuated
parallel robots, the experiment is designed such that the method is
applied to redundant CDR. The Jacobian matrix may be rearranged into
the following form:
\begin{equation}
\label{31}
\begin{array}{c}
J^T=-
\begin{bmatrix}
x-x_{A1} & x-x_{A2} & x-x_{A3} & x-x_{A4} \\
y-y_{A1} & y-y_{A2} & y-y_{A3} & y-y_{A4} \\
z-z_{A1} & z-z_{A2} & z-z_{A3} & z-z_{A4}
\end{bmatrix}\\
\begin{bmatrix}
\frac{1}{l_1} & 0 & 0 & 0 \\
0 & \frac{1}{l_2} & 0 & 0 \\
0 & 0 & \frac{1}{l_3} & 0 \\
0 & 0 & 0 & \frac{1}{l_4}
\end{bmatrix}.
\end{array}
\end{equation}
Thus, Jacobian matrix is expressed in the form $J^T= J^T_{new}(X)L^{-1}$.
Now it is possible to express $J_{new}^T$ in regressor form
\begin{equation}
\begin{array}{c}
\label{35.5}
J_{new}^T=
\begin{bmatrix}
x & x & x & x \\
y & y & y & y \\
z & z & z & z
\end{bmatrix}
+
\begin{bmatrix}
-1 & 0 & 0 \\
0 & -1 & 0 \\
0 & 0 & -1
\end{bmatrix}\times\\
\begin{bmatrix}
x_{A1} & x_{A2} & x_{A3} & x_{A4} \\
y_{A1} & y_{A2} & y_{A3} & y_{A4} \\
z_{A1} & z_{A2} & z_{A3} & z_{A4}
\end{bmatrix}
=
\begin{bmatrix}
x & x & x & x \\
y & y & y & y\\
z & z & z & z
\end{bmatrix}

+Y\Theta.
\end{array}
\end{equation}
$Y_a\theta_a,Y_b\theta_b$ and $Y_c\theta_c$ are determined in Appendix.


\begin{figure}[t!]
    \centering
    \subfigure[Various components of ARAS suspended cable robot]{
        \includegraphics[width=.85\linewidth,height=5cm]{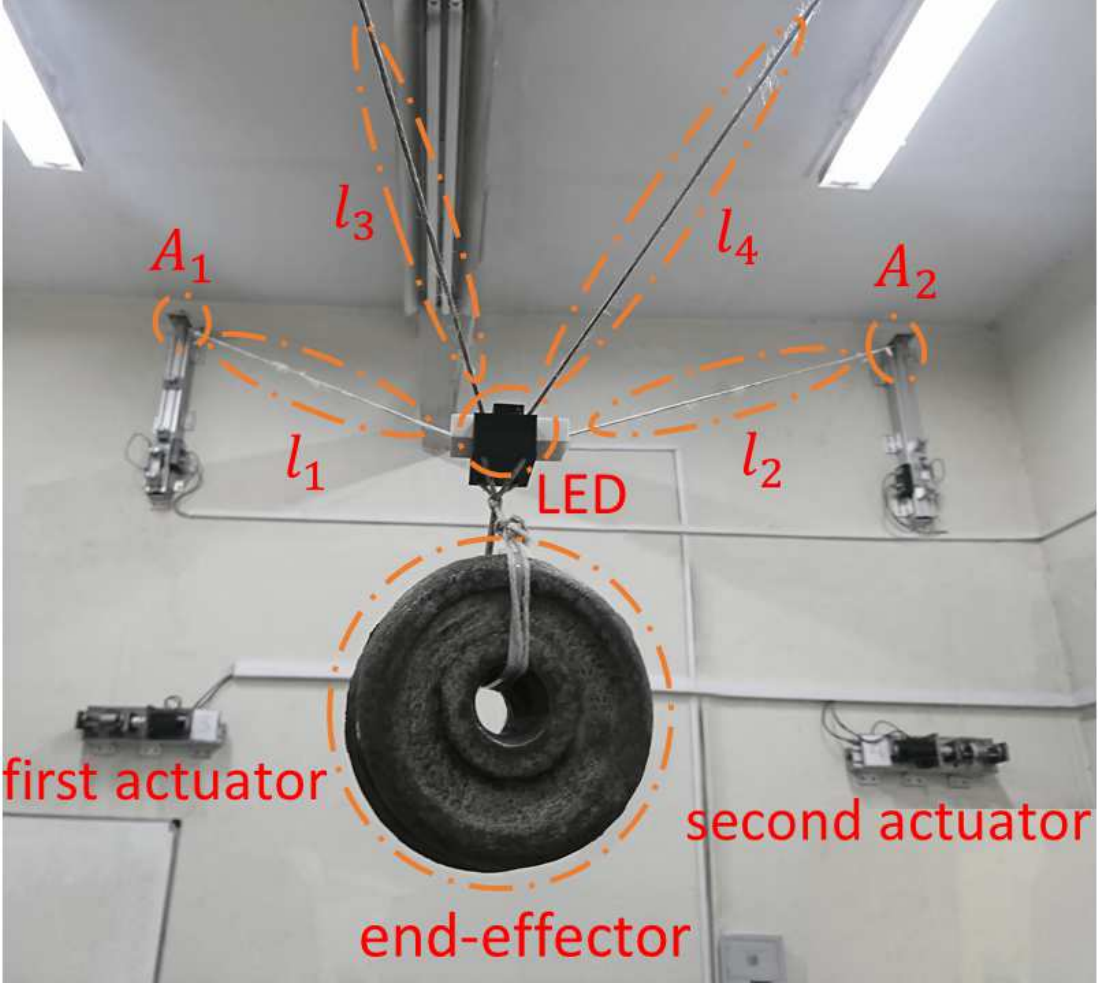}
    }
    \subfigure[The prototype of a stereo vision system, which is attached to the ceiling of laboratory in the geometrical center of the top of the workspace.]{
       \includegraphics[width=.85\linewidth,height=3cm]{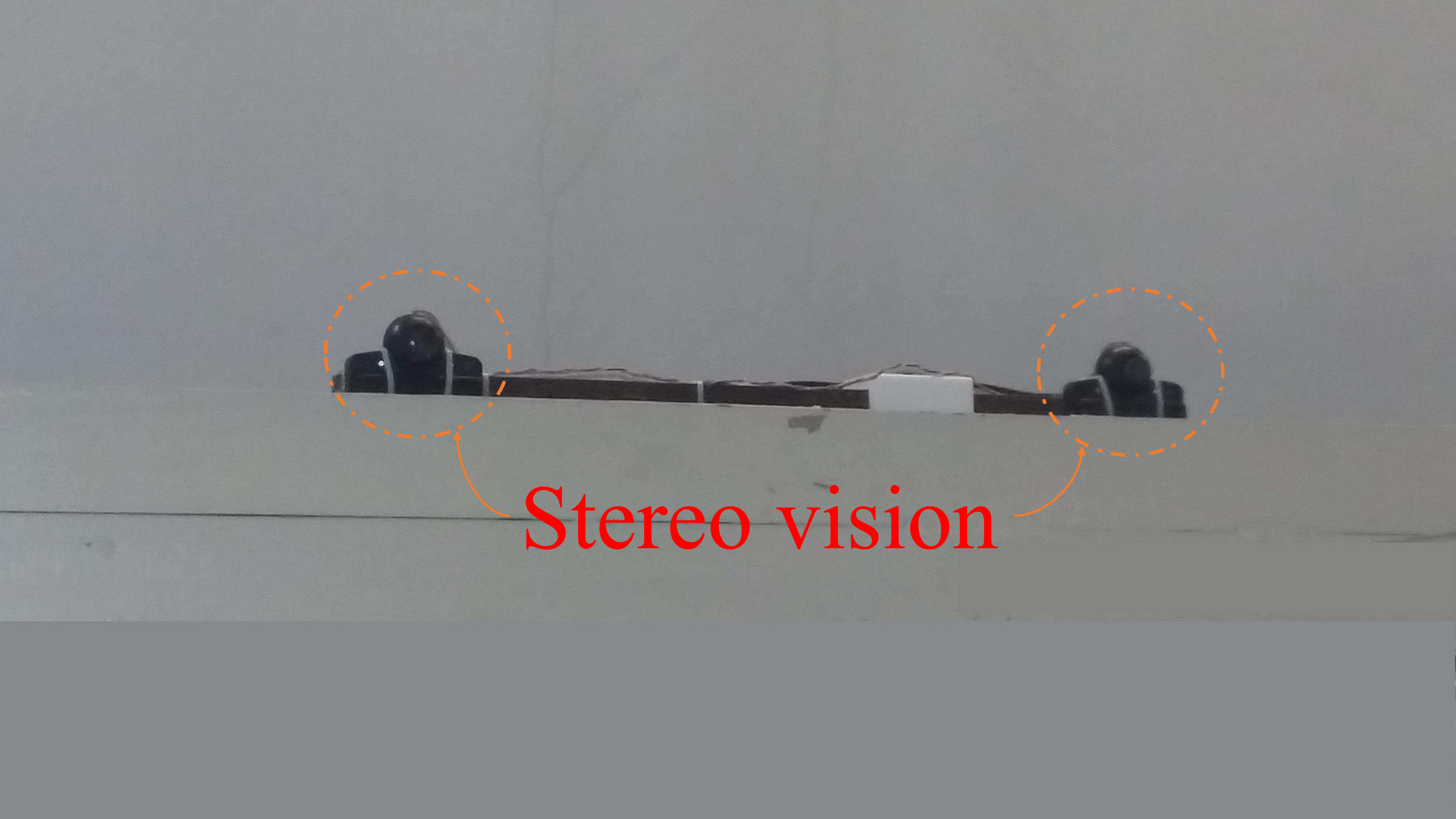}
    }
    \caption{ARAS suspended cable driven robot}
    \label{p3}
\end{figure}
\begin{figure*}[th]
    \centering
    \subfigure[
Path of the robot with adaptive and non-adaptive controllers.
     In contrast to non-adaptive response which is based on calibrated model, the path with adaptive controller is almost matched with the
    desired path after a transient response.]{
        \includegraphics[width=.45\linewidth,height=5.3cm]{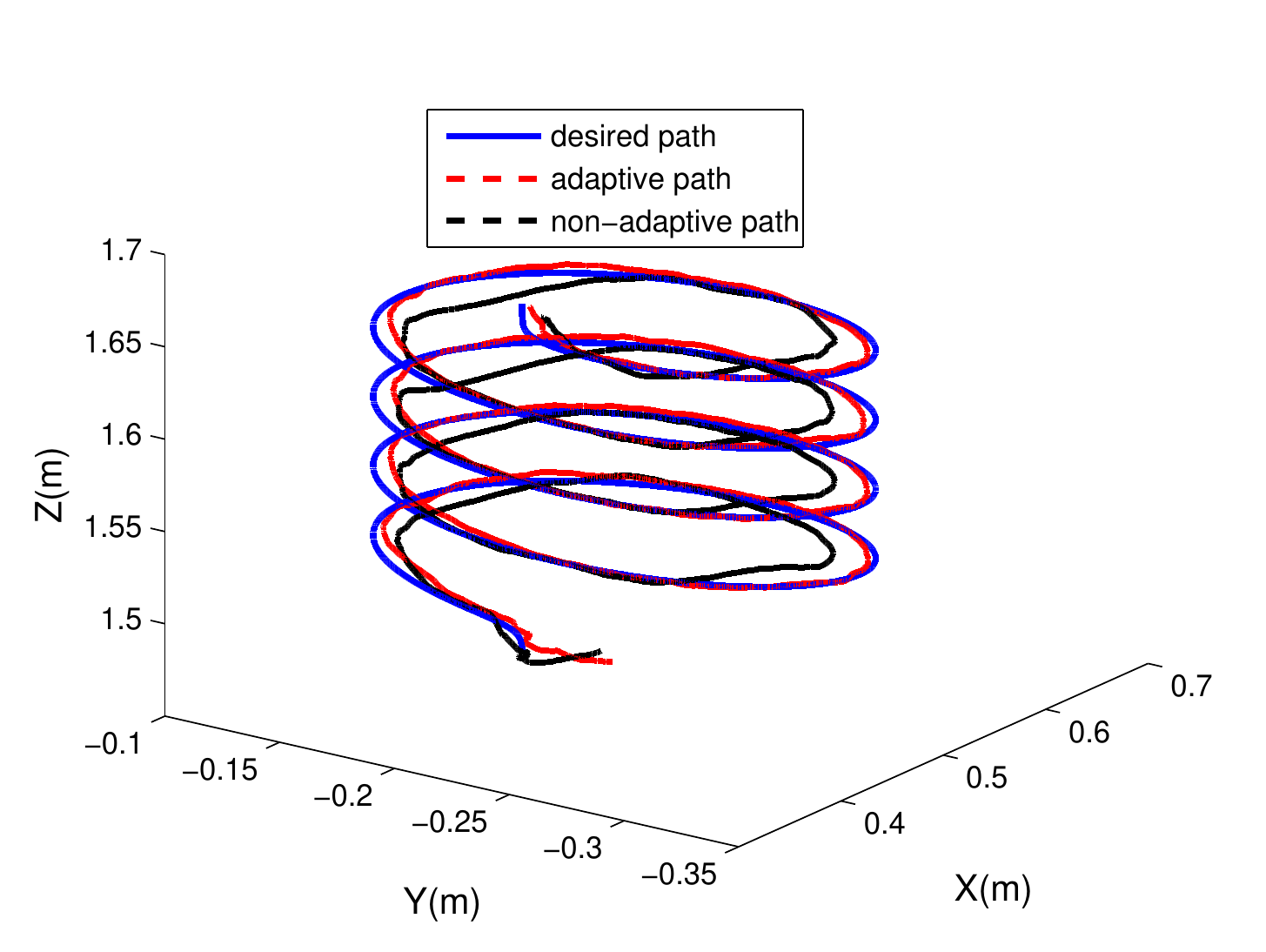}
        \label{p21}
    }
    \hspace*{1mm}
    \subfigure[
Tracking error of $x$ in centimeter.
    Maximum remaining error with adaptive controller
    is less than 0.5cm which shows superiority of proposed
    method compared to the method based on calibration.]{
        \includegraphics[width=.45\linewidth,height=5.3cm]{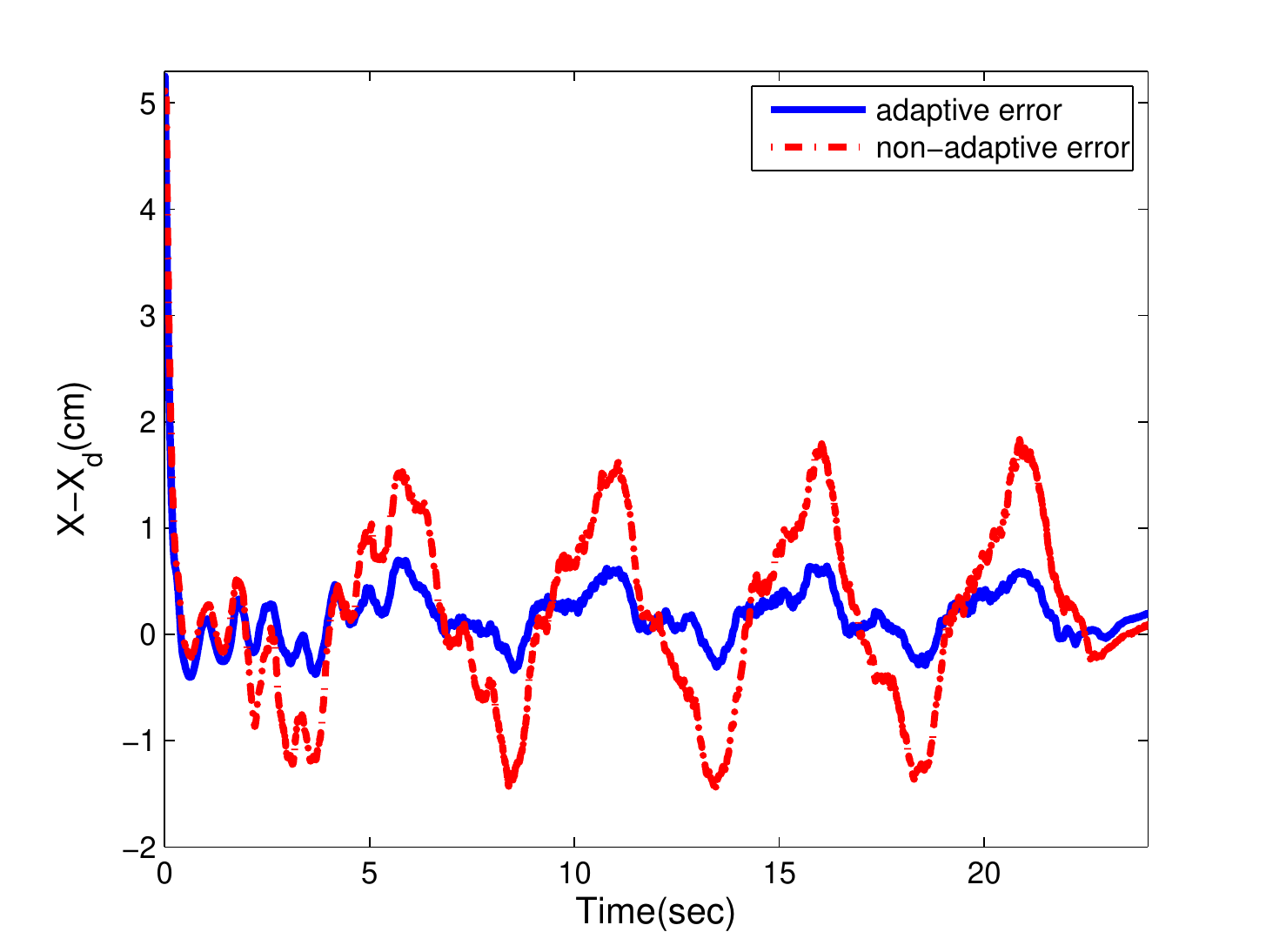}
        \label{p22}
    }
    \subfigure[
Tracking error of $y$ in centimeter.
    Maximum remaining error with adaptive controller
    is about 0.5cm which shows superiority of proposed
    method compared to the method based on calibration.  ]{
        \includegraphics[width=.45\linewidth,height=5.3cm]{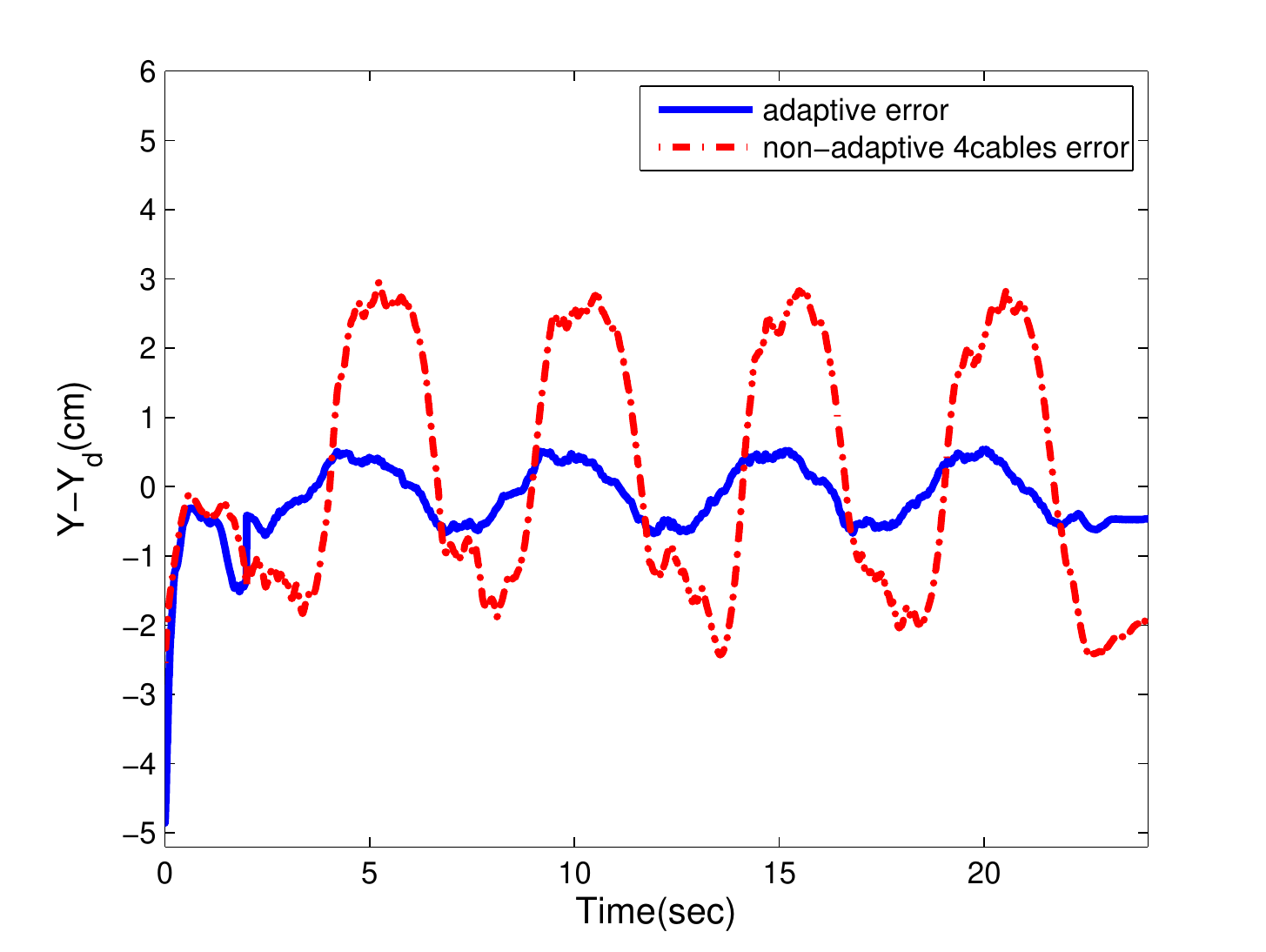}
        \label{p23}
    }
    \hspace*{5mm}
    \subfigure[
Tracking error of $z$ in centimeter.
    Maximum remaining error with adaptive controller
    is about 0.25cm which shows great response of proposed
    methods compared to the method based on calibration.  ]{
        \includegraphics[width=.45\linewidth,height=5.3cm]{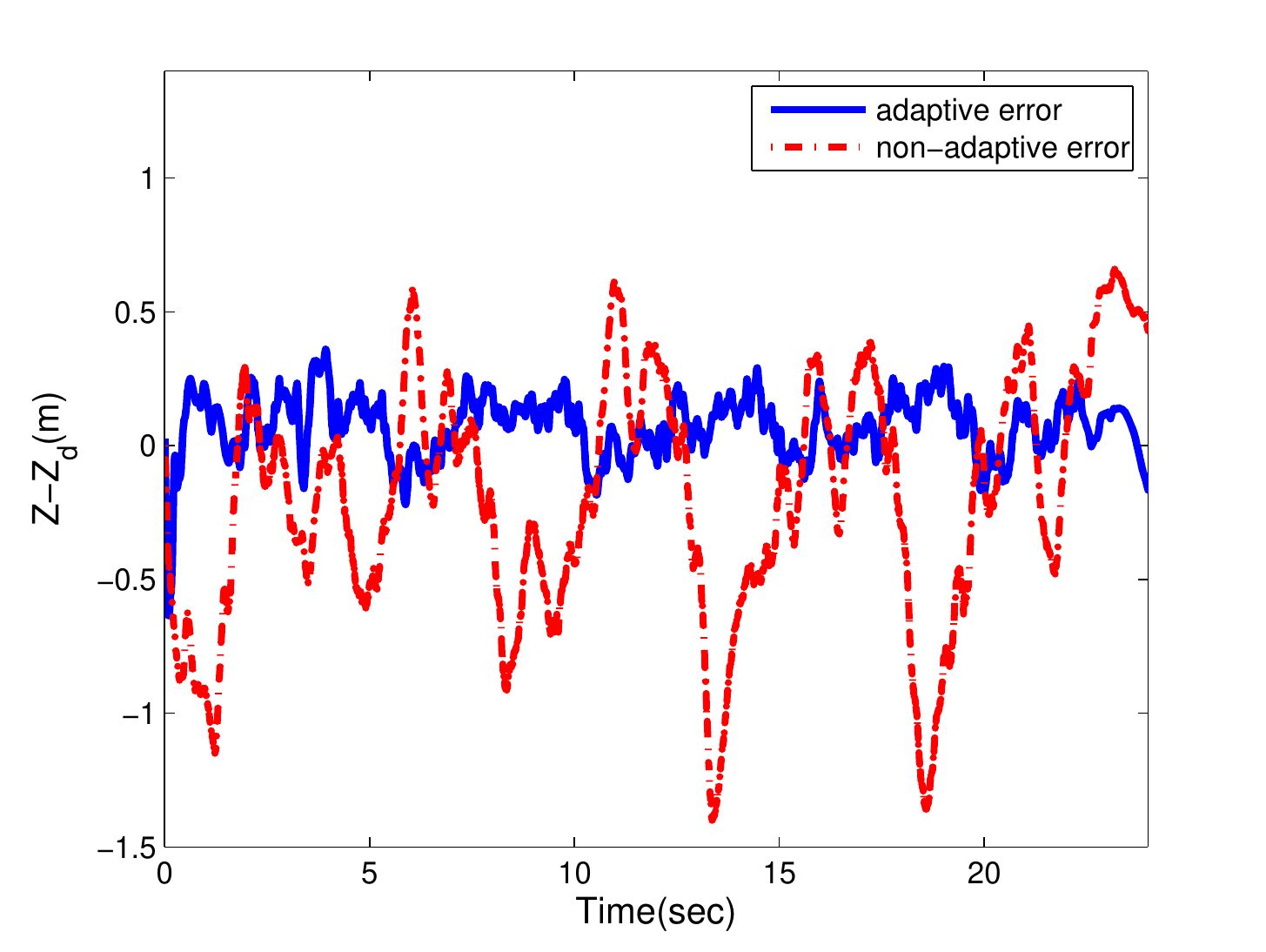}
        \label{p24}
    }
    \caption{Experiment results of the adaptive and non-adaptive controllers on a 3--DOF CDR. Tracking errors with proposed controller is smaller than the non-adaptive controller based on calibrated model. }
    \label{p2}
\end{figure*}
%
\begin{figure}[h]
	\centering
	\includegraphics[width=.99\linewidth,height=5.7cm]{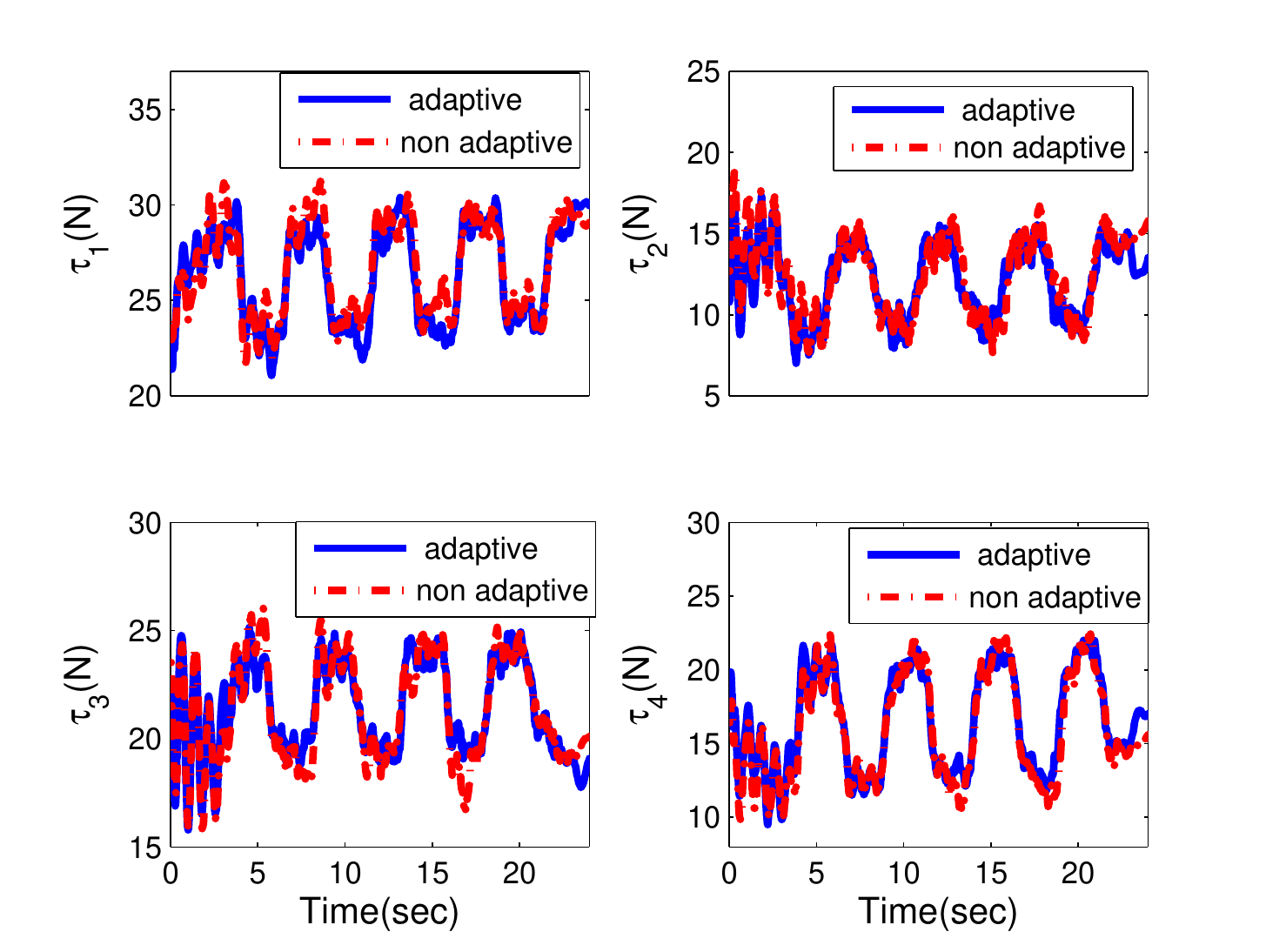}
	\caption{Control efforts of adaptive and non-adaptive controllers in the experiment.  The reason for oscillation
		at the initial moments are fluctuations in cables. The control laws with proposed method have more small fluctuations due to adaptation laws.}
	\label{p4}
\end{figure}
    \begin{figure*}[h!]
	\begin{align}
	&\prescript{}{1}{Y'}_{c}=[x^2\ddot{X}_{r1}/l_1^2+xy\ddot{X}_{r2}/l_1^2+\ddot{X}_{r1}x^2/l_2^2+xy\ddot{X}_{r2}/l_2^2,+y^2\ddot{X}_{r1}/l_1^4
	-xy\ddot{X}_{r2}/l_1^4+y^2\ddot{X}_{r1}/l_2^4-xy\ddot{X}_{r2}/l_2^4-2xy^2\dot{X}_{r1}/l_1^6\nonumber\\&\resizebox{0.99\hsize}{!}{$-2y^3\dot{X}_{r1}/l_1^6+2x^2y\dot{X}_{r2}/l_1^6
		+2xy^2\dot{X}_{r2}/l_1^6-2xy^2\dot{X}_{r1}/l_2^6-2y^3\dot{X}_{r1}/l_2^6+2x^2y\dot{X}_{r2}/l_2^6+2xy^2\dot{X}_{r2}/l_2^6,
		2y^2\ddot{X}_{r1}/l_1^4
		-2xy\ddot{X}_{r2}/l_1^4$}\nonumber\\&\resizebox{0.99\hsize}{!}{$+2y^2\ddot{X}_{r1}/l_2^4-2xy\ddot{X}_{r2}/l_2^4
		-4xy^2\dot{X}_{r1}/l_1^6-4y^3\dot{X}_{r1}/l_1^6+4x^2y\dot{X}_{r2}/l_1^6
		+4xy^2\dot{X}_{r2}/l_1^6-4xy^2\dot{X}_{r1}/l_2^6-4y^3\dot{X}_{r1}/l_2^6+4x^2y\dot{X}_{r2}/l_2^6$}\nonumber\\&\resizebox{0.99\hsize}{!}{$+4xy^2\dot{X}_{r2}/l_2^6,
		2xy^2\dot{X}_{r1}/l_1^5+2y^3\dot{X}_{r1}/l_1^5-2x^2y\dot{X}_{r2}/l_1^5-2xy^2\dot{X}_{r2}/l_1^5+2xy^2\dot{X}_{r1}/l_2^5+2y^3\dot{X}_{r1}/l_2^5
		-2x^2y\dot{X}_{r2}/l_2^5-2xy^2\dot{X}_{r2}/l_2^5$}\nonumber\\&+xy^2\dot{x}\dot{X}_{r1}/l_1^5-x^2y\dot{y}\dot{X}_{r1}/l_1^5-x^2y\dot{x}\dot{X}_{r2}/l_1^5
	+x^3\dot{y}\dot{X}_{r2}/l_1^5+xy^2\dot{x}\dot{X}_{r1}/l_2^5-x^2y\dot{y}\dot{X}_{r1}/l_2^5-x^2y\dot{x}\dot{X}_{r2}/l_2^5+x^3\dot{y}\dot{X}_{r2}/l_2^5,\ddot{X}_{r1},\nonumber\\
	&-2x\ddot{X}_{r1}/l_2^2-y\ddot{X}_{r2}/l_2^2,+y\ddot{X}_{r2}/l_2^4+2y^2\dot{X}_{r1}/l_2^6-4xy\dot{X}_{r2}/l_2^6,
	+2y\ddot{X}_{r2}/l_2^4+4y^2\dot{X}_{r1}/l_2^6-8xy\dot{X}_{r2}/l_2^6,-2y^2\dot{X}_{r1}/l_2^5\nonumber\\
	&+4xy\dot{X}_{r2}/l_2^5-y^2\dot{x}\dot{X}_{r1}/l_2^5+2xy\dot{y}\dot{X}_{r1}/l_2^5+2xy\dot{x}\dot{X}_{r2}/l_2^5+3x^2\dot{y}\dot{X}_{r2}/l_2^5,\ddot{X}_{r1}/l_2^2,2y\dot{X}_{r2}/l_2^6,4y\dot{X}_{r2}/l_2^6,
	-2y\dot{X}_{r2}/l_2^5\nonumber\\&-y\dot{y}\dot{X}_{r1}/l_2^5-y\dot{x}\dot{X}_{r2}/l_2^5+3x\dot{y}\dot{X}_{r2}/l_2^5,0,0,0,-\dot{y}\dot{X}_{r2}/l_2^5], \nonumber\\
	&\prescript{}{2}{Y'}_{c}=[xy\ddot{X}_{r1}/l_1^2+y^2\ddot{X}_{r2}/l_1^2
	+xy\ddot{X}_{r1}/l_2^2+y^2\ddot{X}_{r2}/l_2^2+gx^2/l_1^2+gy^2/l_1^2+gx^2/l_2^2
	+gy^2/l_2^2,-xy\ddot{X}_{r1}/l_1^4
	+x^2\ddot{X}_{r2}/l_1^4\nonumber\\&\resizebox{0.99\hsize}{!}{$-xy\ddot{X}_{r1}/l_2^4+x^2\ddot{X}_{r2}/l_2^4+2x^2y\dot{X}_{r1}/l_1^6+2xy^2\dot{X}_{r1}/l_1^6
		-2x^3\dot{X}_{r2}/l_1^6-2x^2y\dot{X}_{r2}/l_1^6+2x^2y\dot{X}_{r1}/l_2^6+2xy^2\dot{X}_{r1}/l_2^6-2x^3\dot{X}_{r2}/l_2^6$}\nonumber\\&\resizebox{0.99\hsize}{!}{$-2x^2y\dot{X}_{r2}/l_2^6,
		-2xy\ddot{X}_{r1}/l_1^4
		+2x^2\ddot{X}_{r2}/l_1^4-2xy\ddot{X}_{r1}/l_2^4+2x^2\ddot{X}_{r2}/l_2^4+4x^2y\dot{X}_{r1}/l_1^6+4xy^2\dot{X}_{r1}/l_1^6
		-4x^3\dot{X}_{r2}/l_1^6
		-4x^2y\dot{X}_{r2}/l_1^6$}\nonumber\\&+4x^2y\dot{X}_{r1}/l_2^6+4xy^2\dot{X}_{r1}/l_2^6-4x^3\dot{X}_{r2}/l_2^6-4x^2y\dot{X}_{r2}/l_2^6,
	-2x^2y\dot{X}_{r1}/l_1^5-2xy^2\dot{X}_{r1}/l_1^5+2x^3\dot{X}_{r2}/l_1^5+2x^2y\dot{X}_{r2}/l_1^5\nonumber\\&-2x^2y\dot{X}_{r1}/l_2^5
	-2xy^2\dot{X}_{r1}/l_2^5+2x^3\dot{X}_{r2}/l_2^5+2x^2y\dot{X}_{r2}/l_2^5+y^3\dot{x}\dot{X}_{r1}/l_1^5-xy^2\dot{y}\dot{X}_{r1}/l_1^5-xy^2\dot{x}\dot{X}_{r2}/l_1^5+x^2y\dot{y}\dot{X}_{r2}/l_1^5\nonumber\\&+y^3\dot{x}\dot{X}_{r1}/l_2^5
	-x^2y\dot{y}\dot{X}_{r1}/l_2^5-xy^2\dot{x}\dot{X}_{r2}/l_2^5-x^2y\dot{y}\dot{X}_{r2}/l_2^5,\ddot{X}_{r2}+g,
	-y\ddot{X}_{1}/l_2^2-g2x/l_2^2,y\ddot{X}_{1}/l_2^4-2x\ddot{X}_{2}/l_2^4-4xy\dot{X}_{1}/l_2^6\nonumber\\&-2y^2\dot{X}_{1}/l_2^6-6x^2\dot{X}_{2}/l_2^6+4xy\dot{X}_{2}/l_2^6,
	2y\ddot{X}_{1}/l_2^4-4x\ddot{X}_{2}/l_2^4-8xy\dot{X}_{1}/l_2^6-4y^2\dot{X}_{1}/l_2^6-12x^2\dot{X}_{2}/l_2^6+8xy\dot{X}_{2}/l_2^6,4xy\dot{X}_{r1}/l_2^5\nonumber\\&
	+2y^2\dot{X}_{r1}/l_2^5+6x^2\dot{X}_{r2}/l_2^5-4xy\dot{X}_{r2}/l_2^5+y^2\dot{y}\dot{X}_{r1}/l_2^5+y^2\dot{x}\dot{X}_{r2}/l_2^5+2xy\dot{y}\dot{X}_{r2}/l_2^5,
	g/l_2^2,\ddot{X}_{r2}/l_2^4+2y\dot{X}_{r1}/l_2^6+6x\dot{X}_{r2}/l_2^6,\nonumber\\&2g/l_2^2,2\ddot{X}_{r2}/l_2^4+4y\dot{X}_{r1}/l_2^6+12x\dot{X}_{r2}/l_2^6,
	-2y\dot{X}_{r1}/l_2^5-6x\dot{X}_{r2}/l_2^5+2y\dot{X}_{r2}/l_2^5-y\dot{y}\dot{X}_{r2}/l_2^5,0,2\dot{X}_{r2}/l_2^5,4\dot{X}_{r2}/l_2^5,2\dot{X}_{r2}/l_2^5].\label{ap1}\\
	&\hspace{8.5cm}\mathclap{\rule{17cm}{0.4pt}}\nonumber
	\end{align}
\end{figure*}
In order to measure the length of cables, the motor rotation angles
are measured by incremental encoders. Hence, the current length of
cables are available by knowing initial length of the them.  A 100
frame per second stereo vision camera with  $640\times 480$
resolution is utilized to measure position of the LED lamp as the
position of the end-effector. More information about the
experimental setup is given in \cite{khalilpour2019robust}.
 Fig.~\ref{p3} shows different
parts of  ARAS cable driven suspended robot.

The mass of end-effector is equal to 4.5KG and coordinates of cable
anchor points are obtained by calibration as:
\begin{equation}
\label{37.5}
\begin{split}
&x_{A1}=-x_{A2}=x_{A3}=-x_{A4}=\frac{b}{2}=\frac{3.56}{2} \\
&y_{A1}=y_{A2}=-y_{A3}=-y_{A4}=\frac{a}{2}=\frac{7.05}{2} \\
&z_{A1}=z_{A2}=z_{A3}=z_{A4}=h=4.26
\end{split}
\end{equation}
The spring-like desired trajectory is expressed in SI
unit systems, as follows:
 \begin{equation}
\label{37}
\begin{cases}
    x_d(t)=0.48-0.1\cos(\frac{2\pi}{5}t)       \\
    y_d(t)=-0.22+0.1\sin(\frac{2\pi}{5}t) \\
    z_d(t)=1.5+0.0075t
  \end{cases}
  \end{equation}
The center and diameter of the trajectory are chosen in such a way
that the robot is inside its workspace away from its singular
points, and well-measured by the stereo camera.
The adaptive passivity based method parameters which is applied to
redundant case are set to:
\begin{equation*}
\Gamma=20 I,\qquad K=10 I, \qquad \Lambda=5 I.
\end{equation*}
The initial position of the robot is:
\begin{equation*}
\begin{bmatrix}
x_0 \\ y_0 \\ z_0
\end{bmatrix}
=
\begin{bmatrix}
0.43 \\ -0.28 \\ 1.5
\end{bmatrix}.
\end{equation*}
Notice that in contrast to all previous works on ARAS CDR, in this
work the initial position of the robot is not on the trajectory,
i.e. $\tilde{X}$ is not zero at $t=0$. Note that such sudden motion
request in cable driven robots may lead to longitudinal and
transverse oscillations in cables which may cause instabilities in
the robot. This extreme scenario is tested on the robot with
suitable controller performance.

The upper bound of perturbation for dynamic and kinematic parameters
is set to $10\%$. In order to examine the effect of the projection
algorithm, a saturation function is used as a simple appropriate
projection for the case of passivity based method. By this means,
estimated parameters are saturated within the $\pm 15\%$ bounds. For
the sake of comparison, and in order to analyze the performance of
proposed methods, a non-adaptive controller is also
implemented in practice. The control law is as what given in
(\ref{10}) with the parameters obtained from calibration. This is
considered, since a calibrated model is not match exactly with
nominal model of a robot. Note that a high gain 
controller was also implemented on the robot, whose results are not
reported in this paper, since it led to instability due to the high
oscillations in cables.

The experimental results are illustrated in Fig.~\ref{p2} and
Fig.~\ref{p4}. Performance of the controllers are depicted in
Fig.~\ref{p2}. As it is seen in Fig.~\ref{p21}, the traversed path
with adaptive controller suitably tracks the desired path with a
short transient error. However, non-adaptive controller is not that
precise, and leads to an apparent error throughout the path. In
order to compare the results more clearly, the tracking errors are
shown in Fig.~\ref{p22}, \ref{p23} and \ref{p24}.

Note that the tracking error illustrated in these figures are in
centimeters. The results show the desirable  performance of the
proposed method in comparison to non-adaptive controller based on
the calibrated model. The response of the system is affected by the
oscillations of cables at initial transient due to an initial error
between the trajectory and position of the robot. After this period,
fluctuations are suitably damped and thus, the robot has almost a
repetitive response.

In Fig.~\ref{p22}, the tracking error in $X$ direction is less than
$0.5$cm with adaptive controller while with the non-adaptive
controller, it is about 2cm. The tracking errors in $Y$ and $Z$
directions are about 0.5cm and 0.25cm for proposed method and 3cm
and 1.5cm with the non-adaptive controller, respectively. This shows
superiority of the proposed methods compared to that of current
available method, since the response is improved and bounds of the
errors are decreased. Notice that the reason why error in $Y$
direction is almost double of that in $X$ direction is the distance
between anchor points proposed in (\ref{37.5}). Recall that in the
case of non-adaptive controller, the kinematic and dynamic
parameters are obtained based on a time consuming calibration.
Indeed, if the parameters were unknown, a worse response or even
instability, would be happen. Note that the non-vanishing error may
be caused by a simple dynamic model assigned to the robot and
dynamics of the actuators.

Fig.~\ref{p4} shows control efforts for adaptive and non-adaptive
controllers in experiments. As it is seen in this figure, some
oscillations are observed at the initial moments. The main reason
for such oscillations are the oscillations caused in the cables,
because of its elasticity, while the reason why control laws with
proposed method have smaller oscillations is the adaptation law.
Note that all control signals are positive, since as explained in
Remark~\ref{r1}, the desired trajectory is within the feasible
workspace of the robot as well as using the projection algorithm, it
is ensured that adapted parameters can not exceed from a specified
bound. Finally as it is depicted in this  figure, the control
efforts needed in the proposed adaptive controller are almost
similar to that of non-adaptive controller, despite their suitable
tracking performance.
    \begin{figure*}[t]
	\begin{align}
	&\prescript{}{1}{Y}_{a}=[-2z^2\ddot{X}_{r1}-2xz(\ddot{X}_{r3}+g),2z^2\big(\prescript{}{1}{K}S\big)+2xz\big(\prescript{}{3}{K}S\big),
	4z\ddot{X}_{r1}+2x(\ddot{X}_{r3}+g),-4z\big(\prescript{}{1}{K}S\big)-2x\big(\prescript{}{3}{K}S\big),-2\ddot{X}_{r1},\nonumber
	\\&2\resizebox{0.99\hsize}{!}{$\big(\prescript{}{1}{K}S\big),-2z^2\ddot{X}_{r2}+yz(\ddot{X}_{r3}+g),2z^2\big(\prescript{}{2}{K}S\big)-yz\big(\prescript{}{3}{K}S\big),4z\ddot{X}_{r2}
		-y(\ddot{X}_{r3}+g),-4z\big(\prescript{}{2}{K}S\big)+y\big(\prescript{}{3}{K}S\big)-2\ddot{X}_{r2},2\big(\prescript{}{2}{K}S\big),$}\nonumber\\&\resizebox{0.99\hsize}{!}{$(4x^2z+4z^3)(\ddot{X}_{r3}+g),-(4x^2z+4z^3)\big(\prescript{}{3}{K}S\big),-(4x^2+12z^2)(\ddot{X}_{r3}+g),(4x^2+12z^2)\big(\prescript{}{3}{K}S\big),2y^2z(\ddot{X}_{r3}+g),-2y^2z\big(\prescript{}{3}{K}S\big),$}\nonumber\\&x-2y^2(\ddot{X}_{r3}+g),2y^2\big(\prescript{}{3}{K}S\big),12z(\ddot{X}_{r3}+g),-12z\big(\prescript{}{3}{K}S\big),-4(\ddot{X}_{r3}+g),4\big(\prescript{}{3}{K}S\big)],\nonumber\\
	&\prescript{}{2}{Y}_{a}=[2z^2\ddot{X}_{r1}+2xz(\ddot{X}_{r3}+g),-2z^2\big(\prescript{}{1}{K}S\big)-2xz\big(\prescript{}{3}{K}S\big),v4z\big(\prescript{}{1}{K}S\big)+2x\big(\prescript{}{3}{K}S\big),-4z\ddot{X}_{r1}-2x(\ddot{X}_{r3}+g),2\ddot{X}_{r1},\nonumber\\&
	\resizebox{0.99\hsize}{!}{$-2\big(\prescript{}{1}{K}S\big),-2z^2\ddot{X}_{r2}+yz(\ddot{X}_{r3}+g),2z^2\big(\prescript{}{2}{K}S\big)-yz\big(\prescript{}{3}{K}S\big),4z\ddot{X}_{r2}-y(\ddot{X}_{r3}+g)-4z\big(\prescript{}{2}{K}S\big)+y\big(\prescript{}{3}{K}S\big)-2\ddot{X}_{r2},2\big(\prescript{}{2}{K}S\big),$}\nonumber\\&\resizebox{0.99\hsize}{!}{$(4x^2z+4z^3)(\ddot{X}_{r3}+g),-(4x^2z+4z^3)\big(\prescript{}{3}{K}S\big),-(4x^2+12z^2)(\ddot{X}_{r3}+g),(4x^2+12z^2)\big(\prescript{}{3}{K}S\big),2y^2z(\ddot{X}_{r3}+g),-2y^2z\big(\prescript{}{3}{K}S\big),$}\nonumber\\&-2y^2(\ddot{X}_{r3}+g),2y^2\big(\prescript{}{3}{K}S\big),12z(\ddot{X}_{r3}+g),-12z\big(\prescript{}{3}{K}S\big),-4(\ddot{X}_{r3}+g),4\big(\prescript{}{3}{K}S\big)],\nonumber \\
	&\resizebox{0.99\hsize}{!}{$\prescript{}{3}{Y}_{a}=[-2z^2\ddot{X}_{r1}-2xz(\ddot{X}_{r3}+g),2z^2\big(\prescript{}{1}{K}S\big)+2xz\big(\prescript{}{3}{K}S\big),-4z\big(\prescript{}{1}{K}S\big)-2x\big(\prescript{}{3}{K}S\big),-2\ddot{X}_{r1},
		4z\ddot{X}_{r1}+2x(\ddot{X}_{r3}+g),\big(\prescript{}{1}{K}S\big),$}\nonumber\\&22z^2\ddot{X}_{r2}+yz(\ddot{X}_{r3}+g),-2z^2\big(\prescript{}{2}{K}S\big)+yz\big(\prescript{}{3}{K}S\big),-4z\ddot{X}_{r2}+y(\ddot{X}_{r3}+g),4z\big(\prescript{}{2}{K}S\big)-y\big(\prescript{}{3}{K}S\big),2\ddot{X}_{r2},-2\big(\prescript{}{2}{K}S\big),(4x^2z\nonumber\\&+4z^3)(\ddot{X}_{r3}+g),-(4x^2z+4z^3)\big(\prescript{}{3}{K}S\big),-(4x^2+12z^2)(\ddot{X}_{r3}+g),(4x^2+12z^2)\big(\prescript{}{3}{K}S\big),2y^2z(\ddot{X}_{r3}+g),-2y^2z\big(\prescript{}{3}{K}S\big),\nonumber\\&-2y^2(\ddot{X}_{r3}+g),2y^2\big(\prescript{}{3}{K}S\big),12z(\ddot{X}_{r3}+g),-12z\big(\prescript{}{3}{K}S\big),-4(\ddot{X}_{r3}+g),4\big(\prescript{}{3}{K}S\big)],\nonumber\\
	&\prescript{}{4}{Y}_{a}=[2z^2\ddot{X}_{r1}+2xz(\ddot{X}_{r3}+g),-2z^2\big(\prescript{}{1}{K}S\big)-2xz\big(\prescript{}{3}{K}S\big),
	-4z\ddot{X}_{r1}-2x(\ddot{X}_{r3}+g),4z\big(\prescript{}{1}{K}S\big)+2x\big(\prescript{}{3}{K}S\big),2\ddot{X}_{r1},\nonumber\\&\resizebox{0.99\hsize}{!}{$-2\big(\prescript{}{1}{K}S\big),2z^2\ddot{X}_{r2}+yz(\ddot{X}_{r3}+g),-2z^2\big(\prescript{}{2}{K}S\big)+yz\big(\prescript{}{3}{K}S\big),-4z\ddot{X}_{r2}+y(\ddot{X}_{r3}+g),4z\big(\prescript{}{2}{K}S\big)-y\big(\prescript{}{3}{K}S\big),2\ddot{X}_{r2},-2\big(\prescript{}{2}{K}S\big),$}\nonumber\\&\resizebox{0.99\hsize}{!}{$(4x^2z+4z^3)(\ddot{X}_{r3}+g),-(4x^2z+4z^3)\big(\prescript{}{3}{K}S\big),-(4x^2+12z^2)(\ddot{X}_{r3}+g),(4x^2+12z^2)\big(\prescript{}{3}{K}S\big),2y^2z(\ddot{X}_{r3}+g),-2y^2z\big(\prescript{}{3}{K}S\big),$}\nonumber\\&-2y^2(\ddot{X}_{r3}+g),2y^2\big(\prescript{}{3}{K}S\big),12z(\ddot{X}_{r3}+g),-12z\big(\prescript{}{3}{K}S\big),-4(\ddot{X}_{r3}+g),4\big(\prescript{}{3}{K}S\big)].\label{ap2}\\
	&\hspace{8.5cm}\mathclap{\rule{17cm}{0.4pt}}\nonumber
	\end{align}
\end{figure*}
\section{Conclusions and Prospect Research}\label{s5}
This paper focused on the design of adaptive tracking controller for
parallel robots with dynamic and kinematic uncertainties. A novel
expression for inverse of Jacobian matrix in regressor form was
proposed, a methods based on passivity was introduced and adaptation
law for unknown parameters was elicited. By this means, it was
proved that the tracking error of the robot converges asymptotically
to zero in the presence of kinematics, Jacobian and dynamic
uncertainties. The performance of the controller was verified
through simulation and experiment, and it has been shown that in
comparison to available methods, the response is improved, while the
effect of projection in singularity avoidance was highlighted. Since
the research on the control of parallel robots in presence of
kinematic and dynamic uncertainties is developing, future research
may be devoted to decoupling the adaptation laws for kinematic and
dynamic parameters in order to reduce the number of adapting parameters. Extension of the proposed method to the case of serial robots is also underway.

\appendix
\section*{Regressor Forms of 2R\underline{P}R robot}
Jacobian matrix can be represented as:
\begin{align*}
J^T&=\begin{bmatrix}
\frac{x}{l_1} & \frac{x-a}{l_2} \\ \frac{y}{l_1}
& \frac{y}{l_2}
\end{bmatrix}=\begin{bmatrix}
x & x-a \\
y & y
\end{bmatrix}
\begin{bmatrix}
\frac{1}{l_1} & 0 \\
0 & \frac{1}{l_2}
\end{bmatrix}=J_{new}^TL^{-1}
\end{align*}
\begin{align*}
J_{new}^T=\begin{bmatrix}
x & x \\
y & y
\end{bmatrix}+\begin{bmatrix}
0 & -1 \\
0 & 0
\end{bmatrix}
a=\begin{bmatrix}
x & x \\
y & y
\end{bmatrix}+Y\Theta
\end{align*}
$$T=ay=Y_b\theta_b.$$
Dynamic parameters of the system can be expressed by a regressor as
follows 
$$Y_c\theta_c=-KS+Y'_C\theta'_c,$$
where $Y_c$ is proposed in (\ref{ap1}) and $\theta'_c$ is:
$$\theta'_c=[m,I,mc^2,mc,m_p,ma,Ia,mc^2a,mca,ma^2,Ia^2,mc^2a^2,$$
$$mca^2,ma^3,Ia^3,mc^2a^3,mca^3]^T\in\mathbb{R}^{17}.$$
$Y_a\theta_a$ proposed in (\ref{14}) is obtained as follows
$$Y_a=[R_1Y_c\vdots0_{4\times2}]+[0_{4\times2}\vdots Y_c],\quad R_1:=\begin{bmatrix}
y & -x \\ -y & x
\end{bmatrix},$$
$$ \theta_a=[\theta_c,a^4m,a^4I,a^4mc^2,a^4mc]^T\in\mathbb{R}^{21}.$$

\section*{Regressor Forms of Redundant CDR}
 Consider, $a,b$ and $h$ as illustrated in Fig.~\ref{sh1}
and presented in (\ref{37.5}).
 Matrix $R$ proposed in (\ref{11}) is
\begin{equation*}
\resizebox{.95\hsize}{!}{$
\begin{bmatrix}
-2a^2b(z-h)^2 & -2ab^2(z-h)^2 & \begin{array}{c} 2(2a^2x^2-a^2bx-\\2b^2y^2+ab^2y)(z-h)\\+4a^2(z-h)^3.\end{array} \\
2a^2b(z-h)^2 & -2ab^2(z-h)^2 & \begin{array}{c} 2(2a^2x^2+a^2bx-\\2b^2y^2+ab^2y)(z-h)\\+4a^2(z-h)^3.\end{array} \\
-2a^2b(z-h)^2 & 2ab^2(z-h)^2 & \begin{array}{c} 2(2a^2x^2-a^2bx-\\2b^2y^2-ab^2y)(z-h)\\+4a^2(z-h)^3.\end{array} \\
2a^2b(z-h)^2 & 2ab^2(z-h)^2 & \begin{array}{c} 2(2a^2x^2+a^2bx-\\2b^2y^2-ab^2y)(z-h)\\+4a^2(z-h)^3.\end{array}
\end{bmatrix}
$}
\end{equation*}
 $Y_a$ is proposed in (\ref{ap2}) and $\theta_a$ is given by
\begin{equation*}
\begin{split}
&\theta_a=[a^2bm,a^2b,a^2bhm,a^2bh,a^2bh^2m,a^2bh^2,ab^2m,ab^2,\\&ab^2hm,ab^2h,ab^2h^2m,ab^2h^2,a^2m,a^2,a^2hm,a^2h,b^2m,b^2,\\&b^2hm,b^2h,a^2h^2m,a^2h^2,a^2h^3m,a^2h^3]^T.
\end{split}
\end{equation*}
The determinant $T=Y_b\theta_b$ is
\begin{equation*}
\begin{split}
&Y_b=[4z^2,-8z,4,-128(x^2yz^2+xy^2z^2),256(x^2yz^2\\&+xy^2z^2)+128xyz^2,-128(x^2yz^2+xy^2z^2),256xy],
\end{split}
\end{equation*}
\begin{equation*}
\theta_b=[a^2b^2,a^2b^2h,a^2b^2h^2,h,h^2,h^3,h^4]^T.
\end{equation*}
$Y_c\theta_c$ is given as
\begin{equation}
  \begin{array}{c}
M\ddot{X}_r+C\dot{X}_r+G-KS=-KS\\+\left(\ddot{X}_r+
\begin{bmatrix}
  0 \\ 0 \\ g
\end{bmatrix}
\right )m=-KS+Y_c\theta_c.
  \end{array}
\end{equation}
 Construction of $Y_F\theta_F$ is a bit lengthy but straightforward.

\bibliographystyle{ieeetr}
\bibliography{REF1}
\end{document}